\documentclass[runningheads,orivec]{llncs}
\usepackage[T1]{fontenc}
\usepackage{pifont}
\usepackage{siunitx}
\usepackage{graphicx}
\usepackage{todonotes}
\usepackage[dvipsnames]{xcolor}
\RequirePackage{listings}
\RequirePackage{ amssymb }
\RequirePackage{colortbl}
\RequirePackage{ amsmath }
\RequirePackage{ dsfont }
\usepackage{xspace}
\usepackage{bigfoot}
\usepackage{hyperref}
\usepackage{xurl}
\usepackage[labelfont=bf, figurename=Figure]{caption}
\usepackage{subcaption}
\usepackage{lstautogobble}
\usepackage{zref-savepos}
\usepackage{chngcntr}
\usepackage{paralist}
\input{prelude.tex}
\begin{document}

\counterwithout{lstlisting}{chapter}
\title{Scalable Deductive Verification of Data-Level Parallel Programs}

\author{Lars B. van den {Haak}\inst{1}\orcidID{0000-0002-0330-5016} \and
Anton {Wijs}\inst{1}\orcidID{0000-0002-2071-9624} \and
Marieke {Huisman}\inst{2}\orcidID{0000-0003-4467-072X}
}
\authorrunning{L. B. van den Haak, A.J. Wijs and M. Huisman}
\institute{
Eindhoven University of Technology, The Netherlands \\
\email{\{l.b.v.d.haak, a.j.wijs\}@tue.nl} \and
University of Twente, The Netherlands \\
\email{m.huisman@utwente.nl}
}
\maketitle              % typeset the header of the contribution
\begin{abstract}
This paper introduces several techniques that improve the scalability of the deductive verification of data-level programs working on arrays and matrices. First of all, we introduce a technique to rewrite expressions with (nested) quantifiers, so suitable triggers can be generated for these expressions. We have proven this rewrite technique correct in a theorem prover. Second, we make reasoning about potentially overlapping arrays easier, by providing specification constructs to indicate and verify that two arrays are not aliases, or that they are immutable, so they can be modelled as mathematical sequences. All our techniques are implemented in the \vercors program verifier. We illustrate how our techniques improve scalability through a large number of experiments. Using our techniques on a set of typical GPU kernels, we achieve a reduction of verification time by, on average, a factor of 9, with outliers being up to 150 times faster. Additionally, applying these techniques to earlier experiments and an earlier case study of a radio telescope pipeline permitted the verification of results which were previously unobtainable and significantly reduced the verification time.

\keywords{Deductive Verification \and Separation Logic \and Parallel \and GPU \and GPGPU \and Quantifiers}
\end{abstract}

%!TEX root = main.tex

\section{Introduction}\label{sec:intro}

\emph{Data} (or \emph{Data-level}) \emph{parallel programming} is an approach to parallel programming
in which multiple threads concurrently perform the same operations on different data elements~\cite{hillisDataParallelAlgorithms1986}.
Over the years, this approach to parallel programming has become increasingly popular, in particular in the form
of the \emph{Single Instruction, Multiple Data} (SIMD) paradigm. Modern \emph{Graphics Processing
Units} (GPUs) offer a prominent way of SIMD programming, but also \emph{Central Processing Units} (CPUs) support
SIMD instructions. GPU computing, and therefore data parallel programming, has had a major impact on scientific
computing in fields such as computational biology (e.g., genomics)~\cite{WieSpri12}, statistics~\cite{LiuTan12}, and physics (e.g., fluid dynamics)~\cite{BerBet11}. Furthermore, GPUs effectively accelerate computations involving matrix-vector and
matrix-matrix multiplications~\cite{GreLok11,WijsBos12}, and
they made an enormous impact on Artificial Intelligence with deep learning~\cite{deeplearning}.

% is an emerging paradigm.\todo{Lars: Can we still call it emerging? Is not well established nowadays? Also I think "data parallel" is more common than "data-level parallel".} It supports
%parallelisation across multiple processes in parallel computing
%environments, such that each process operates on its own part of the
%data. There is hardware tailored to support this specific
%programming model, such as GPU's, and this programming model is becoming more prominent.
%
Due to this success, these days, there
is a plethora of programming languages supporting this paradigm,
such as SYCL, OpenCL, CUDA and OpenMP. Typical data structures that
are used in these programs are arrays and matrices.

%There are many possible applications, ranging from physics to chemistry and
%biology, etc. Correctness is important here\todo{of course, this needs
%  to be have some evidence, better argumentation}, and therefore
%, deductive verification support is being developed to reason about this
%class of programs.

However, deductive verification of these programs, i.e., verifying that they
satisfy formal contracts, often runs into scalability
issues. Verification of a computation quickly takes (too) long
or becomes unfeasible. Bottlenecks are in the reasoning 
at the SMT level: the deductive verification tool applies program logics
that create SMT proof obligations.
We have identified multiple causes for this:
\begin{compactitem}
  \item Generated proof obligations often contain (multiple)
    nested quantifier expressions, and SMT solvers have difficulty
    reasoning about them as they often involve \emph{triggers}, i.e., patterns,
    that are not suitable for actually triggering an instantiation of the quantifiers.
    
    \item If a program uses multiple arrays, matrices, or higher dimensional array structures,
    then these could potentially be aliases of each other. This needs to be verified, as it
    may influence program correctness, but that leads to many additional
      proof obligations.

    \end{compactitem}

%\todo[caption={}]{
%\begin{enumerate}
%    \item Title, Abstract and Introduction: 2 pages
%    \item Background: 1 page
%    \item Quantifiers: 4 pages
%    \item Unique types: 5 pages
%    \item Experiments: 3 pages
%    \item Related Work: 1 page
%    \item Discussion and Future Work: 1 page
%    \item \textbf{total:} 17 pages (1 page to spare)
%    \item \textbf{CAV limit:} 18 pages
%\end{enumerate}
%}

    We introduce multiple solutions to address these causes, all of which are
 implemented in the deductive program verifier \vercors~\cite{armborstVerCorsVerifierProgress2024}. \vercors uses
 permission-based separation logic as its specification language, which means that for all shared
    memory, the user needs to specify permission annotations that
    capture whether a thread has read or write access at a certain
    moment. Moreover, all functional properties are framed, i.e. one
    can only state something about that part of the memory to which one
    has access.

    First of all, after the background has been explained in Section~\ref{sec:background},
    we propose a rewriting procedure in Section~\ref{sec:quantifiers} to
    flatten expressions with nested quantifiers into equivalent
    expressions with single quantifiers, for which triggers can be
    generated that make the verification process efficient. We have proven 
    the correctness of this procedure using \lean.

    Second, in Section~\ref{sec:types}, we propose a mechanism to indicate that
    arrays (and other data structures) are not aliases of each other. For this, we
    use \emph{uniqueness types}. This helps to significantly reduce
    the reasoning about potential aliases.

    % We do this soundly in cases 
    % This helps to significantly reduce the reasoning 
    % about potential aliases,
    
    % to indicate that we do not have 
    % to consider arrays as aliases in the verification context. This implies that
    % permissions for such array

    % arrays (and other data structures) are not aliases. For this, we
    % use uniqueness types, that capture that there is only a single
    % reference to a certain array. This helps to significantly reduce
    % the reasoning about potential aliases.
    % }

    Sometimes, multiple arrays may be aliases of each other, but they are never updated
    in the program. For this case, we propose to indicate
    that the arrays are immutable. By doing so, the deductive verifier can reason about
    them as if they are \emph{sequences}.
%    , where we do need to provide
%    and reason about permission specifications. 

%    Finally, we combine this with a recently implemented feature of
%    VerCors, called \textit{extract}, which we have further enhanced in this work for GPU kernels. This helps us break down the
%    verification of a larger GPU kernel into smaller parts. (Section~\ref{sec:extract})

    Furthermore, in Section~\ref{sec:experiments}, we discuss a large number of experiments that demonstrate
    how these
    features together speed up verification and verify previously unverifiable cases of a Radio Telescope Pipeline case study, including an implementation of the Padre algorithm~\cite{vandenhaakVerifyingRadioTelescope2024}.
    
    % With these optimisations, we could significantly reduce verification times and verify previously unverifiable cases of a Radio Telescope Pipeline case study~\cite{vandenhaakVerifyingRadioTelescope2024}.
    % \todo{This would be a point to
      % mention Padre} %(Section~\ref{sec:experiments})
    
    Finally, related work is discussed in Section~\ref{sec:related}, and conclusions are drawn in Section~\ref{sec:concl}.
    The data for the experiments (Section~\ref{sec:experiments}), the \lean proof (Section~\ref{sec:quantifiers}), and the version of the \vercors tool used in this paper can be found in an accompanying artefact at \url{http://github.com/cav2026-anonymous/Scalable-Deductive-Verification-of-Data-Level-Parallel-Programs}.

%!TEX root = main.tex

\section{Background}\label{sec:background}
Deductive verification tools use formal contracts, consisting of pre- and post-conditions, written as \emph{requires} and \emph{ensures} statements, respectively, to formally verify a program. To support concurrent programs,
the \vercors~\cite{armborstVerCorsVerifierProgress2024} tool uses first-order logic, enhanced with (concurrent) separation logic concepts~\cite{brookesSemanticsConcurrentSeparation2004}. For each memory location (or \textit{heap} location) \pvlinlineb|a|, one needs to specify a permission to access it: \pvlinlineb|Perm(a, p)|, with \pvlinlineb|p| a fractional number between \pvlinlineb|0| and \pvlinlineb|1|, where \pvlinlineb|p==1| indicates that we can write to location \pvlinlineb|a| and \pvlinlineb|0<p<1| indicates read permission. A permission of \pvlinlineb|1| is often written as \pvlinlineb|write|. Correct permission annotations allow \vercors to prove memory safety.
%In \vercors, most program statements are eventually modelled using \textit{inhale} and \textit{exhale} statements, constructs similar to \textit{assert} and \textit{assume}. However, they can also remove or add permission. For instance, the assignment \pvlinlineb|a=5| is modelled as \pvlinlineb|inhale Perm(a, write); exhale Perm(a,write) && a==5|, where \vercors first removes write permission, which fails if it is not available, then adds it back again and assumes that \pvlinlineb|a==5|.
To verify a program, \vercors translates it into \viper's~\cite{mullerViperVerificationInfrastructure2016,eilersFifteenYearsViper2025} intermediate verification language. In turn, \viper queries SMT solvers to solve the proof obligations.

%Besides normal functions, \vercors also has \pvlinlineb|pure| functions. These functions can depend on the heap but cannot have side-effects; i.e., change the contents of the heap. Since there are no side effects, these can be used in annotations. The function \pvlinlineb|pure int abs(int x) = x>0 ? x : -x;| is a typical example that calculates the absolute value.
%\todo{Lars: I think I've added everything which is strictly needed to understand the contents of the rest of the sections. Alter as you see fit, if you want to specialise it more to GPUs or introduce the GPU examples explained below.}

%\todo{maybe use this example to illustrate VerCors, as it shows why it's so important to handle this case. \\
%Lars: Sure sounds like a good idea. There is space left next to it, so we could introduce another example next to it, which does verify (as the current listing does not verify without the changes in this paper).}

\begin{lstlisting}[float=t,style=pvl2, caption={An \OpenCL GPU kernel to swap the contents of two arrays, inspired by the \pvlinlineb|Xswap| kernel from CLBlast~\cite{nugterenCLBlastTunedOpenCL2018}.},label={lst:swap-cl}]
/*@ $\SP$... $\label{ln:swap-contract-1}$
    $\UP$context (\forall* int i; 0 <= i && i < n/get_global_size(0);$\label{ln:swap-write-1}$
      Perm(ygm[\gtid + i*get_global_size(0)], write)) &&
      Perm(xgm[\gtid + i*get_global_size(0)], write)); $\label{ln:swap-write-2}$
    $\UP$ensures (\forall int i; 0 <= i && i < n/get_global_size(0); $\label{ln:swap-ensure-1}$
      ygm[\gtid + i*get_global_size(0)] == \old(xgm[\gtid + i*get_global_size(0)]) &&
      xgm[\gtid + i*get_global_size(0)] == \old(ygm[\gtid + i*get_global_size(0)])); $\label{ln:swap-ensure-2}$
@*/ $\label{ln:swap-contract-2}$
// Swap the contents of two arrays in parallel.
__kernel void XswapFast(const int n, __global float* xgm, __global float* ygm) {
  /*@ $\SP$loop_invariant \gtid <= id && id < n + get_global_size(0);
      $\UP$loop_invariant (\forall* int i; 0 <= i && i < n/get_global_size(0);
        Perm(ygm[\gtid + i*get_global_size(0)], write)) &&
        Perm(xgm[\gtid + i*get_global_size(0)], write));
      $\UP$loop_invariant (\forall int i; id/get_global_size(0) <= i && i < n/get_global_size(0);
        ygm[\gtid + i*get_global_size(0)] == \old(ygm[\gtid + i*get_global_size(0)]) &&
        xgm[\gtid + i*get_global_size(0)] == \old(xgm[\gtid + i*get_global_size(0)]));
      $\UP$loop_invariant (\forall int i; 0 <= i && i < id/get_global_size(0);
        ygm[\gtid + i*get_global_size(0)] == \old(xgm[\gtid + i*get_global_size(0)]) &&
        xgm[\gtid + i*get_global_size(0)] == \old(ygm[\gtid + i*get_global_size(0)]));
  @*/
  for (int id = get_global_id(0); id < n; id += get_global_size(0)) { $\label{ln:swap-grid-1}$
    float temp = xgm[id];
    xgm[id] = ygm[id];
    ygm[id] = temp;
  } $\label{ln:swap-grid-2}$
}
\end{lstlisting}

Listing~\ref{lst:swap-cl} presents an example of an \OpenCL GPU kernel that can be used to swap the contents of two arrays
\pvlinlineb|xgm| and \pvlinlineb|ygm|, both of size \pvlinlineb|n|, in parallel.
The special comments \pvlinlineb|/*@| and \pvlinlineb|@*/| indicate annotations. 
The contract for this function addresses
that the array contents are indeed swapped: at lines \ref{ln:swap-write-1}--\ref{ln:swap-write-2} (l.\ref{ln:swap-write-1}--\ref{ln:swap-write-2}), the \emph{context} is given that the kernel has write permission for both
arrays (\pvlinlineb|context| refers to the fact that this holds both before and after executing the kernel), and at l.\ref{ln:swap-ensure-1}--\ref{ln:swap-ensure-2}, the postcondition
is given that the array contents have been swapped, with \pvlinlineb|\old| referring to the old contents of the arrays
before the kernel was launched.
A number of keywords and functions are used to refer to the thread hierarchy on a GPU.
Using CUDA terminology, threads are executed in \emph{blocks} of a predefined size, and a predefined number of
blocks make up a \emph{grid}. Both blocks and grids can be one-, two-, or three-dimensional. In this example, we use only
one dimension. We use the following functions, with \pvlinlineb|i| referring to a dimension ($0 \leq \pvm{i} \leq 2$):
\begin{compactitem}
\item \pvlinlineb|get_local_size(i)| returns the size of a thread block in the \pvlinlineb|i|-dimension.
\item \pvlinlineb|get_num_blocks(i)| returns the total number of blocks in the \pvlinlineb|i|-dimension.
\item \pvlinlineb|get_local_id(i)| returns the (block-local) ID of the current thread in the \pvlinlinec|i|-dimension. Note that
$0 \leq$ \pvlinlineb|get_local_id(i)| $<$ \pvlinlineb|get_local_size(i)|.
\item \pvlinlineb|get_block_id(i)| returns the ID of the block in which the thread resides in the \pvlinlineb|i|-dimension. Note that
$0 \leq$ \pvlinlineb|get_block_id(i)| $<$ \pvlinlineb|get_num_blocks(i)|.
\item \pvlinlineb|get_global_id(i)| returns the \emph{global} ID of the thread in the \pvlinlineb|i|-dimension,
defined as
\pvlinlineb|get_local_size(i)| $\cdot$ \pvlinlineb|get_block_id(i)| $+$ \pvlinlineb|get_local_id(i)|.
\item \pvlinlineb|get_global_size(i)| returns the total number of threads in the \pvlinlineb|i|-dimension, defined
as \pvlinlineb|get_num_blocks(i)| $\cdot$ \pvlinlineb|get_local_size(i)|.
\end{compactitem}

Finally, \pvlinlineb|\gtid| is equal to \pvlinlineb|get_global_id(0)| for the one dimensional case.

Note that these functions are used to divide the work among the threads: Initially, every thread accesses those
elements in \pvlinlineb|xgm| and \pvlinlineb|ygm| located at index \pvlinlineb|get_global_id(0)|. As the arrays can
be larger than the total number of threads, the for-loop at l.\ref{ln:swap-grid-1}--\ref{ln:swap-grid-2} is actually a \emph{grid-stride} loop, which is very common in GPU kernels: after accessing
their initial element, every thread jumps \pvlinlineb|get_global_size(0)| positions ahead in the arrays to access the
next elements, and this continues until the end of the arrays has been reached. To reason about the loop, note the
\emph{loop invariants} at l.11--21 that also use the functions mentioned above. These are essential for \vercors
to prove the contract of the kernel.

\begin{lstlisting}[float=t,style=pvl2, caption={A GPU PVL kernel to swap the contents of two arrays, with the number of threads being at least as large as the number of elements.},label={lst:swap-pvl}]
/*@ $\SP$...
    $\UP$ensures \forall int i; 0 <= i && i < n*m; xgm[i] == \old(ygm[i]) && ygm[i] == \old(xgm[i]);
@*/
void XswapFastSmallArrays(const int n, float* xgm, float* ygm, int nb, int nt) {
  par blocks (int bid = 0 .. nb)
  {
    /*@ $\SP$...
         $\UP$ensures \forall int i; 0 <= i && i < n; xgm[bid*nt + i] == \old(ygm[bid*nt + i]);
         $\UP$ensures \forall int i; 0 <= i && i < n; ygm[bid*nt + i] == \old(xgm[bid*nt + i]); @*/
    par threads (int tid = 0 .. nt)
      /*@ $\SP$...
          $\UP$ensures xgm[bid*nt + tid] == \old(ygm[bid*nt + tid]);
          $\UP$ensures ygm[bid*nt + tid] == \old(xgm[bid*nt + tid]); @*/
    {
      float temp = xgm[bid*nt + tid];
      xgm[tid] = ygm[bid*nt + tid];
      ygm[tid] = temp;
    }
  }
}
\end{lstlisting}

\vercors encodes parallel functions, such as GPU kernels, using \emph{parallel blocks}.
%According to Blom et al. (Theorem 3), to check the correctness of a
To verify a GPU kernel, one needs to quantify over all thread blocks in a grid and all threads in each thread block~\cite{blomSpecificationVerificationGPGPU2014},
resulting in two nested parallel blocks. We have made this explicit in Listing~\ref{lst:swap-pvl}, using \vercors'
intermediate language \pvl. For convenience, the grid-stride loop has been removed. Note that the contract for the
parallel block labelled \pvlinlineb|blocks| quantifies over all threads in a block, and that the contract for the kernel
in turn quantifies over all threads in the grid. Although given here explicitly, \vercors typically derives these contracts
 automatically
from a given contract for a thread block, which is possible due to the very structured way in which GPU threads operate.
For the kernel in Listing~\ref{lst:swap-cl}, \vercors would also internally generate these contracts, given the contract
at l.\ref{ln:swap-contract-1}--\ref{ln:swap-contract-2}.

\paragraph*{Triggers.}
Contracts for deductive verification frequently involve \emph{quantifiers}, for instance in a statement such as
\pvlinlineb{\forall int i; 0<=i< n; A[i]== 0}. During verification, a deductive verification tool needs to apply such a
statement whenever applicable, for instance when encountering the ground term \pvlinlineb|A[1]| in the code: the quantifier must be
\emph{instantiated}.  Mapping \pvlinlineb|1| on \pvlinlineb|i| is straightforward, but for more complex expressions,
heuristics are typically required, although they often lead to incompleteness of the verification. Therefore, \vercors relies on pattern-based quantifier instantiation~\cite{detlefsSimplifyTheoremProver2005} (or E-matching): the user should indicate that an expression
in an annotation should serve as a \emph{trigger} (or \emph{pattern}) for quantification, by enclosing it in \pvlinlineb|{:| and \pvlinlineb|:}|. This trigger (or possibly a set of triggers) should mention all quantified variables. Furthermore, \viper's intermediate verification language does not allow triggers to involve arithmetic expressions.%\footnote{This restriction can be circumvented by introducing additional structure and replacing the arithmetic with functions. However, this precise structure always needs to be present to instantiate the triggers.}.\todo{MH: Deze voetnoot wordt in de volgende paragraaf ook al benoemd. Kan misschien weg?}

For instance, consider the code and contract
in Listing~\ref{lst:triggers-pvl}. At l.\ref{ln:trigger}, \pvlinlineb|f(k,y)| is marked as
a trigger.
The assertion at l.\ref{ln:inst} verifies, as \pvlinlineb|f(x,y)| matches the trigger at l.\ref{ln:trigger}.
However, the assertion at l.\ref{ln:no-inst} fails to verify, as in order to match \pvlinlineb|f(10,z)| needs the information of the quantifier in l.\ref{ln:no-trigger}. However, it is not allowed to specify \pvlinlineb|f(k*i,y)| as a trigger, as it contains arithmetic. Even if we introduce a helper function for \pvlinlineb|*|, such that it can be used as a trigger, it cannot match suitable instances for \pvlinlineb|i| and \pvlinlineb|k| when given \pvlinlineb|f(10)|, as this would need to be written as \pvlinlineb|f(5*2)|, \pvlinlineb|f(2*5)|, \pvlinlineb|f(1*10)| or \pvlinlineb|f(10*1)|.

\begin{lstlisting}[float=t,style=pvl2, caption={A PVL function example to illustrate triggers. Here \pvlinlineb|f| is some uninterpreted side-effect free function.},label={lst:triggers-pvl}, linerange={3-9}]
/*@ $\SP$requires x>0 && y>0;
       $\UP$requires (\forall int k; 0<=k && k<=x ; {: f(k, y) :} == 0); $\label{ln:trigger}$
       $\UP$requires (\forall int k, int i; 0<=k && k<=x && 0<=i && i<=y; f(k*i, z) == 0); @*/ $\label{ln:no-trigger}$
void bar(int x, int y, int z) {
    assert f(x, y) == 0;     $\label{ln:inst}$// this verifies, as "f(x,y)" triggers the first quantifier
    assert f(10, z) == 0;   $\label{ln:no-inst}$// this assertion fails, as no quantifier is triggered
}
\end{lstlisting}
%!TEX root = main.tex

\newcommand{\base}[0]{\mathit{base}}
\newcommand{\off}[0]{\mathit{off}}
\newcommand{\xmin}[0]{\mathit{min}}
\newcommand{\xmax}[0]{\mathit{max}}

\section{Triggers for Nested Quantifiers}\label{sec:quantifiers}

%\paragraph*{The problem.}

As data-level parallel programs typically involve data stored in arrays, their contracts usually contain annotations that
quantify over the data elements. To reason about these, the underlying automated solver used by a deductive verification
tool needs to automatically find instantiations of those quantifiers.
%In general, this is a hard problem, and therefore \textit{triggers} have been introduced so that the user can hint to the prover about an appropriate instantiation~\cite{detlefsSimplifyTheoremProver2005}.
%% In \vercors, a user can write \pvlinline|{:p:}| inside a quantifier to indicate that the expression \pvlinline|p| should be used as a trigger. The underlying solvers will then instantiate the quantifier when an expression with a similar structure of \pvlinline|p| is found.
%Unfortunately, triggers have limitations in how they can be used, in particular when a quantification in the contract does
%not completely match the relevant code.
%%for example when reasoning about flattened multi-dimensional arrays.
To illustrate this, consider Listing~\ref{lst:swap-cl} again, in particular, the postcondition given at l.5--6. We
indicate that the index expression for \pvlinlineb|ygm| should be used as a trigger:
%following expression about a 2-dimensional array $A$.
\begin{lstlisting}[style=pvlcomments, firstline=1,lastline=2, numbers=none]
\forall $\SP$int i; 0 <= i < n/get_global_size(0); {:ygm[\gtid + i*get_global_size(0)]:} ==
$\UP$\old(xgm[\gtid + i*get_global_size(0)])
\end{lstlisting}
When trying to verify the kernel, the solver tries to find suitable instantiations of the quantified variables in the given postcondition, but is not able to do this, as the index expression contains arithmetic
operations.
%Here, \pvlinlineb|A[x_2][x_1]| works as a trigger: when the solver encounters \pvlinlineb|A[2][2]|, this trigger tells it to instantiate the quantifier with \pvlinlineb|x_2 =>> 2| and \pvlinlineb|x_1 =>> 2| to verify its validity.

%However, if we flatten the array, the equivalent expression becomes
%\lstinputlisting[style=pvl2, firstline=8,lastline=8]{Chapters/6-Vercors/quantifiers.pvl}
%where \pvlinlineb|A_f| is the flattened version of \pvlinlineb|A|. Now, we would like to use \pvlinlineb|A_f[4*x_2+ x_1]| as a trigger, but this is not a valid trigger as it contains arithmetic. Moreover, if we would like to use this as a trigger, we would need a way to translate the indexed value in an expression like \pvlinlineb|A_f[10]>0| back to \pvlinlineb|x_1| and \pvlinlineb|x_2|.

The solution that we develop here is to automatically rewrite the quantifier expression in such a way that the arithmetic operations are
removed from the trigger.
We do this by defining a bijective mapping from a linear arithmetic expression containing quantified variables to a single quantified variable, and adding suitable
conditions, such that the validity of the expression remains unchanged. For this concrete example,
assuming for the moment that \pvlinlineb|\gtid| is a given value instead of one calculated using several values,
this gives the following rewritten quantifier.
%\lstinputlisting[style=pvlcomments, firstline=3,lastline=4, numbers=none]{Chapters/6-Vercors/quantifiers1.pvl}
%Our rewriting procedure, which we further explain in this section, rewrites this to the following quantifier.
\begin{lstlisting}[style=pvlcomments, firstline=5,lastline=6, numbers=none]
\forall $\SP$int x; (abs(x-\gtid))%get_global_size(0)==0 && 0 <= x-\gtid < get_global_size(0)*(n/
$\UP$get_global_size(0)); ygm[x] == \old(xgm[x]);
\end{lstlisting}
After the rewriting, \pvlinlineb|ygm[x]| can be used as a trigger. In this section, we explain our rewriting
procedure.
In general, if an index is computed injectively, then we can apply this rewriting by determining an inverse function.
If we have an indexing function and its inverse, we can translate between the original and the rewritten quantifier.

%In this section, we focus on quantifiers over arrays.
%However, the approach can also be used for functions or other data structures accessed in a similar linear way. 

\paragraph{Rewriting quantifiers.}
%Quantifiers that index an array $A$ using index expressions with arithmetic operations should be rewritten.
In general, the quantifiers we wish to rewrite are of the following form, with $\pvm{A}$ an array.
%To define the bijection that transforms an index expression into a single variable, we first precisely define the index patterns for which our quantifier rewriting approach can be applied automatically.
%The properties we are interested in are similar to our previous example, but we generalise the bounds and index expressions, as well as the number of variables that we have.  This results in the following property pattern.
%Consider a generalisation of our previous example without using specific values for bounds and indexing.
% \lstinputlisting[style=pvl2, firstline=71,lastline=71]{Chapters/6-Vercors/quantifiers.pvl}
% This can be rewritten into the following expression
% \lstinputlisting[style=pvl2, firstline=74,lastline=74]{Chapters/6-Vercors/quantifiers.pvl}
% provided that \pvlinlineb|n_1>0|, \pvlinlineb|n_2>0|, \pvlinlineb|a_1>0| and \pvlinlineb|a_2>0| hold, and that \pvlinlineb|a_1*x_1 <= a_2| holds for every value that \pvlinlineb|x_1| can take. In that case, the  inverse for the indexing function \pvlinlineb|(x_1, x_2) =>> a_2*x_2+a_1*x_1| is 
% \pvlinlineb|x =>> (x%a_2, x/a_1)|. 

% Next, we generalise  to an arbitary number of variables.
\begin{lstlisting}[style=pvlcomments, firstline=7,lastline=7, numbers=none]
\forall int x_1, ..., x_k; X(x_1, ..., x_k); R(A[a_k*x_k + ... + a_1*x_1 + b], (x_1, ..., x_k))
\end{lstlisting}

The predicate \pvlinlineb|X(x_1, ..., x_k)| determines the \textit{domain} of the quantifier, i.e., all values of \pvlinlineb|(x_1, ..., x_k)| for which the quantifier holds.
%We refer with $X \subseteq \mathds{Z}^k$ to this domain.
Furthermore,
\pvlinlineb|R| is a predicate with free variables $\pvm{x}_1, \ldots, \pvm{x}_k$ that may involve conditions on the $\pvm{x}_i$, besides the definition of their domains and the index expression.
Based on the index expression in the quantifier form given above, we define a function for array indexing:
\[
f(\pvm{x}_1, \ldots, \pvm{x}_k) = \pvm{a}_k\cdot \pvm{x}_k + \cdots + \pvm{a}_1\cdot \pvm{x}_1 + \pvm{b} = \sum_{i=1}^k \pvm{a}_i \cdot \pvm{x}_i + \pvm{b} 
\]
Reversing this indexing would mean that we could use a statement such as \pvlinlineb|A[x]| as a trigger, and map \pvlinlineb|x|
back to the corresponding \pvlinlineb|(x_1, ..., x_k)| that satisfies $f(\pvm{x_1}, \ldots, \pvm{x_k}) = \pvm{x}$.
To achieve this, we need to define the inverse function $f^{-1}$. For this to work, it must be that each $\pvm{x}_i$ has an (inclusive) lower bound $\xmin_i$ defined in the domain $X$. When the domain also has an (exclusive) upper bound for an $\pvm{x}_i$, we denote this by $\xmax_i$. The upper bound $\xmax_\pvm{k}$ for \pvlinlineb|x_k| must be defined,
to bound the overall quantification. With $n_k = \xmax_k - \xmin_k$, we refer to the size of the $\pvm{x}_k$-dimension. All in all, this means that \pvlinlineb|X(x_1,...,x_k)| should be defined as follows.
\[
\pvm{X(x}_1\pvm{,...,x}_k\pvm{) = C(x}_1\pvm{, ..., x}_k\pvm{)} \wedge \bigwedge_{i=1}^k (\xmin_{i} \leq \pvm{x}_{i}) \wedge
 \pvm{x}_k < \xmax_\pvm{k} 
 %\numberthis
\]
% For the first and second pattern, domain $X$ also has the constraint $\bigwedge_{i=1}^{k} x_i < \xmax_i$, instead of only for $x_k$.
Here, \pvlinlineb|C(x_1,..., x_k)| is a predicate expressing additional constraints on the variables, besides the
constraints for the bounds. With $\pvm{X}$ as above, we wish to rewrite the original quantifier to an equivalent
quantifier with a single variable, of the following form.
\begin{lstlisting}[style=pvlcomments, firstline=8,lastline=8, numbers=none, aboveskip=1ex, belowskip=1ex]
\forall int x; Y(x); R(A[x], $\smash{f^{-1}}$(x))
\end{lstlisting}

Here, the predicate \pvlinlineb|Y(x)| determines the domain of the new quantifier.
%, and \pvlinlineb|A[x]| serves as the
%trigger.
%We refer with $Y \in \mathds{Z}$ to this domain.
To define $f^{-1}$, we first define an offset $\off$ and a helper function $\base_i$.
Intuitively, $\off$ is the value where we start indexing array $A$, which is the position related to the case where
each $\pvm{x}_i$ has its lowest possible value. The function $\base_i$ maps $\pvm{x}$ back to $\pvm{a}_i \cdot \pvm{x}_i$, meaning that we can retrieve $\pvm{x}_i$ by computing $\base_i(\pvm{x})/ |\pvm{a}_i|$.
\begin{align*}
\off &= \sum_{i=1}^k \pvm{a}_i\cdot \xmin_i + \pvm{b} \\ %\label{eq:off} \\
\base_i(\pvm{x}) &= \begin{cases} 
|\pvm{x} - \off| & \text{if } i = k, \\
\base_{i+1}(\pvm{x}) \bmod \pvm{a}_{i+1} & \text{if } 1 \leq i < k.\end{cases} %\label{eq:base}
\end{align*}

Next, we can define the function $f^{-1}$ as follows.
\begin{align*}
f^{-1}(\pvm{x}) = (\base_1(\pvm{x})/ |\pvm{a}_1| + \xmin_1, \ldots, \base_k(\pvm{x})/ |\pvm{a}_k|+ \xmin_k) %\label{eq:finv}
\end{align*}
This rewriting is correct, provided that
\begin{align}
n_k > 0 \numberthis \label{eq:pat3-1}
\end{align}
and the following conditions hold for all
 \(\pvm{a}_i\) ($1 \leq i \leq k$).
\begin{align*}
& \pvm{a}_i \neq 0 \numberthis \label{eq:pat3-2} \\ 
%& n_k > 0 \numberthis \label{eq:pat3-2} \\ 
i < k \Rightarrow &~ \pvm{a}_i \geq 0 \iff \pvm{a}_{i+1} \geq 0 \numberthis \label{eq:pat3-3} \\
i < k \Rightarrow & \forall (x_1, \ldots, x_k) \in X, 
  \sum_{j=1}^{i} |\pvm{a}_j|\cdot (\pvm{x}_j-\xmin_j) < |\pvm{a}_{i+1}| \numberthis \label{eq:pat3-4}
\end{align*}
%Here $\text{sgn}$ is the sign of the number: either being positive or negative.

Finally, we define $\pvm{Y(x)}$ as follows.
\begin{align*}
\pvm{Y(x)} = & \ \pvm{C(}f^{-1}\pvm{(x))} \wedge \base_1(\pvm{x}) \bmod \pvm{a}_1 = 0\ \wedge \\%\numberthis \label{eq:domainY} \
	& (\pvm{a}_1>0 \Rightarrow 0 \leq \pvm{x} - \off < \pvm{a}_k\cdot n_k ) \wedge 
	 (\pvm{a}_1<0 \Rightarrow \pvm{a}_k\cdot n_k < \pvm{x} - \off \leq 0)
\end{align*}

The following theorem states the correctness of our rewriting.
\begin{theorem}\label{thm:eq-quant}
Given a quantifier of the form \pvlinlineb|\forall| \emph{\pvlinlineb|int x_1, ..., x_k; X(x_1, ..., x_k); R(A[a_k*x_k + ... + a_1*x_1 + b], (x_1, ..., x_k))|}, with $k > 0$. If equations~\ref{eq:pat3-1}--\ref{eq:pat3-4} hold, then 
\begin{align*}
\forall (\pvm{x}_1, \ldots, \pvm{x}_k) \in \mathds{Z}^k. \pvm{X}(\pvm{x}_1, \ldots, \pvm{x}_k) &\Rightarrow R(A[f(\pvm{x}_1, \ldots, \pvm{x}_k)], (\pvm{x}_1, \ldots, \pvm{x}_k)) \\
&= \\
\forall \pvm{x} \in \mathds{Z}. \pvm{Y(x)} 
&\Rightarrow R(A[\pvm{x}], f^{-1}(\pvm{x}))
\end{align*}
with $f$ and $f^{-1}$ as defined above.
\end{theorem}
\begin{proofsketch}
A full proof has been written in about 2500 lines of \textsc{Lean 4}~\cite{mouraLean4Theorem2021} code, and can be
found in the accompanying artefact of the current paper. For details on this, see Appendix~\ref{sec:appendix-lean-proof}. Here, we provide a sketch of how this proof is structured. With $X$ and $Y$, we refer to the domains
determined by predicates \pvlinlineb|X(x_1,...,x_k)| and $\pvm{Y(x)}$, respectively.
Initially, the theorem can be proven for the case that \pvlinlineb|C(x_1,...,x_k)| $ = \textit{true}$, $\xmin_i = 0$ for all
$i$ and $\pvm{b} = 0$. This implies that $\off = 0$. In this case, it can be proven that $\forall \pvm{x} \in Y. f(f^{-1}(\pvm{x})) = \pvm{x}$ and that $\forall (\pvm{x}_1,\ldots,\pvm{x}_k) \in X. f^{-1}(f(\pvm{x}_1,\ldots,\pvm{x}_k))= (\pvm{x}_1,\ldots,\pvm{x}_k)$. This proves that $f$ and $f^{-1}$ are inverses on $X$ and
$Y$.

Next, we prove that the image of $f$ under $X$ is contained in $Y$, and vice versa for $f^{-1}$. This, together with
the first part of the proof, proves that $f$ is a bijection between $X$ and $Y$.

Subsequently, this is generalised to the case where $\pvm{b}$, $\off$ and the $\xmin_i$ have arbitrary values. We prove that
the definitions of $f$ and $f^{-1}$ can be constructed from the simpler ones used at the start of the proof,
using function composition. Then, we prove that the function compositions are also inverses and bijections.

Then, we prove that the theorem is still correct for an arbitrary predicate \pvlinlineb|C(x_1,...,x_k)|. As this predicate
introduces the same constraints for both $X$ and $Y$, both domains are restricted in the same way.

Finally, we prove that $f$ and $f^{-1}$ being inverses and $f$ being a bijection between $X$ and $Y$ implies
the equality of the two quantifiers. \hfill\(\qed\)
\end{proofsketch}

Next, we revisit the example taken from Listing~\ref{lst:swap-cl}.

\begin{example}
Consider again the following quantifier.
\begin{lstlisting}[style=pvlcomments, firstline=1,lastline=2, numbers=none]
\forall $\SP$int i; 0 <= i < n/get_global_size(0); {:ygm[\gtid + i*get_global_size(0)]:} ==
$\UP$\old(xgm[\gtid + i*get_global_size(0)])
\end{lstlisting}
Note that it matches the linear quantifier pattern, with $\pvm{b} = \gtid$, $\pvm{x}_1 = \pvm{i}$, $\pvm{a}_1 = $\pvlinlineb|get_global_size(0)|, $\xmin_1 = 0$ and $\xmax_1 = $\pvlinlineb|n/get_global_size(0)|.
Furthermore, there is no $\pvm{C}$-predicate.
From this, we can derive that $\off = \gtid$ and $\base_1(\pvm{x}) = |\pvm{x} - \gtid|$ and
$n_1 = $\pvlinlineb|n/get_global_size(0)|, leading to
$\pvm{Y(x)} = |\pvm{x} - \gtid|\ \bmod\ $\pvlinlineb|get_global_size(0)| $= 0 \wedge 0 \leq \pvm{x} - \gtid <\ $\pvlinlineb|get_global_size(0)| $\cdot$\ \pvlinlineb|(n/get_global_size(0))|. As previously stated, this produces the following
new quantifier.
\begin{lstlisting}[style=pvlcomments, firstline=5,lastline=6, numbers=none]
\forall $\SP$int x; (abs(x-\gtid))%get_global_size(0)==0 && 0 <= x-\gtid < get_global_size(0)*(n/
$\UP$get_global_size(0)); ygm[x] == \old(xgm[x]);
\end{lstlisting}
The second condition may be confusing, but note that with integer division, \pvlinlineb|get_global_size(0)*(n/get_global_size(0))| is not necessarily equal to \pvlinlineb|n|. Furthermore, note that $\pvm{x}$ represents the original
index expression \pvlinlineb|\gtid + i*get_global_size(0)|, hence \pvlinlineb|x - \gtid| equals \pvlinlineb|i*get_global_size(0)|. From this, we can derive that \pvlinlineb|0 <= i*get_global_size(0) < get_global_size(0)*| \pvlinlineb|(n/get_global_size(0))|,
which implies \pvlinlineb|0 <= i < n/get_global_size(0)|, i.e., the condition for $\pvm{i}$ in the original
quantifier.
\end{example}

Besides the use of block and thread IDs for array accessing, another cause for quantifier issues is the
\emph{flattening} of multi-dimensional arrays, such as matrices. To achieve regular memory access patterns, it is
generally a good strategy to store matrices in one-dimensional arrays, sometimes even multiple matrices in
a single array. In the following example, we address
such a situation.

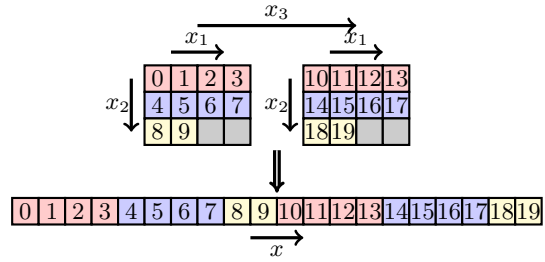
\begin{figure}[t]
\centering
\begin{tikzpicture}[scale=0.35]

\foreach \k in {0,1} {
\foreach \j/\c in {0/red,1/blue}{
  \foreach \i in {0,1,2,3} {
    \filldraw[color=black, fill=\c!20, thick] ($(6*\k+4+\i,-\j)$) rectangle ($(6*\k+4+\i+1,-\j-1)$);
    \pgfmathparse{10*\k+4*\j+\i}
    \draw  ($(6*\k+4+\i+0.5,-\j-0.5)$) node {\pgfmathprintnumber[]{\pgfmathresult}};
  }
 }
}

\foreach \k in {0,1} {
\foreach \j/\c in {2/yellow}{
  \foreach \i in {0,1} {
    \filldraw[color=black, fill=\c!20, thick] ($(6*\k+4+\i,-\j)$) rectangle ($(6*\k+4+\i+1,-\j-1)$);
    \pgfmathparse{10*\k+4*\j+\i}
    \draw  ($(6*\k+4+\i+0.5,-\j-0.5)$) node {\pgfmathprintnumber[]{\pgfmathresult}};
  }
 }
}

\foreach \k in {0,1} {
\foreach \j/\c in {2/black}{
  \foreach \i in {2,3} {
    \filldraw[color=black, fill=\c!20, thick] ($(6*\k+4+\i,-\j)$) rectangle ($(6*\k+4+\i+1,-\j-1)$);
%    \pgfmathparse{10*\k+4*\j+\i}
%    \draw  ($(5*\k+4+\i+0.5,-\j-0.5)$) node {\pgfmathprintnumber[]{\pgfmathresult}};
  }
 }
}

\draw[>=implies, ->,double equal sign distance, very  thick] (9,-3.2) -- (9,-4.8);

\foreach \k in {0,1} {
\foreach \j/\c in {0/red,1/blue}{
  \foreach \i in {0,1,2,3} {
    \filldraw[color=black, fill=\c!20, thick] ($(\k*10+\i+4*\j-1,-5)$) rectangle ($(\k*10+\i+4*\j+1-1,-6)$);
    \pgfmathparse{\k*10+4*\j+\i}
    \draw  ($(\k*10+\i+4*\j+0.5-1,-5.5)$) node {\pgfmathprintnumber[]{\pgfmathresult}};
  }
 }
}

\foreach \k in {0,1} {
\foreach \j/\c in {2/yellow}{
  \foreach \i in {0,1} {
    \filldraw[color=black, fill=\c!20, thick] ($(\k*10+\i+4*\j-1,-5)$) rectangle ($(\k*10+\i+4*\j+1-1,-6)$);
    \pgfmathparse{\k*10+4*\j+\i}
    \draw  ($(\k*10+\i+4*\j+0.5-1,-5.5)$) node {\pgfmathprintnumber[]{\pgfmathresult}};
  }
 }
}

%\draw[->, ultra  thick] (\k*10+6,-3.2) -- (\k*6+6,-4.8);

\foreach \k in {0,1} {
	\draw[->, very thick] (\k*6+5,0.5) -- (\k*6+7,0.5);
	\draw (\k*6+6,1) node {$x_{1}$};
	\draw[->, very thick] (\k*6+3.5,-0.5) -- (\k*6+3.5,-2.5);
	\draw (\k*6+3,-1.5) node {$x_{2}$};
}

\draw[->, very thick] (6,1.5) -- (12,1.5);
\draw (9,2) node {$x_3$};

\draw[->, very thick] (8,-6.5) -- (10,-6.5);
\draw (9,-7) node {$x$};

\end{tikzpicture}
\caption{\label{fig:quant-pat3} An example of storing two matrices together in one array.}
\end{figure}

\begin{example}
Consider Figure~\ref{fig:quant-pat3}, with two matrices stored into a single, one-dimensional array.
A quantifier suitable for the original, non-flattened matrices is the following.
\begin{lstlisting}[style=pvlcomments, firstline=9,lastline=10, numbers=none, aboveskip=1ex, belowskip=1ex]
\forall $\SP$int x_1, int x_2, int x_3; 0<=x_1<4 && 0<=x_2<3 && 0<=x_3<2 && 4*x_2+x_1<10;
$\UP$A[10*x_3+4*x_2+x_1]>x_2;
\end{lstlisting}
Note the constraint \pvlinlineb|4*x_2+x_1<10|, which ensures that the grey cells are not accessed. It serves as the
$\pvm{C}$-constraint in the definition of \pvlinlineb|X(x_1,...,x_k)|. Note also
that every time \pvlinlineb|x_2| is incremented, we jump four elements ahead in the flattened array, corresponding
to the upper-bound of \pvlinlineb|x_1|, and when \pvlinlineb|x_3| is incremented, we jump ahead ten elements,
due to the constraint \pvlinlineb|4*x_2+x_1<10|.
To rewrite this quantifier, first note that $\pvm{a}_1 = 1$, $\pvm{a}_2 = 4$, and $\pvm{a}_3 = 10$.
Furthermore, we have $\off = 10 \cdot 0 + 4 \cdot 0 + 0 + 0 = 0$, $\base_3(\pvm{x}) = \pvm{x} - 0 = \pvm{x}$, $\base_2(\pvm{x}) = \base_3(\pvm{x})\bmod 10 = \pvm{x} \bmod 10$, and $\base_1(\pvm{x}) = \base_2(\pvm{x}) \bmod 4 = (\pvm{x} \bmod 10) \bmod 4$.
%By the definition of $f^{-1}$, this gives us $\pvm{x}_1 = \base_1(\pvm{x}) / \pvm{a}_1 = ((\pvm{x} \bmod 10) \bmod 4) / 1$, $\pvm{x}_2 = \base_2(\pvm{x}) / \pvm{a}_2 = (\pvm{x} \bmod 10)/4$, and $\pvm{x}_3 = \base_3(\pvm{x}) / \pvm{a}_3 = \pvm{x}/10$.
This means that $\pvm{Y(x)}$ is defined as $4\cdot((\pvm{x} \bmod 10)/4)+((\pvm{x} \bmod10) \bmod 4)/1<10 \wedge
((\pvm{x} \bmod 10) \bmod 4) \bmod 1 = 0 \wedge 0 \leq \pvm{x} < 10\cdot 2$.
Note that $((\pvm{x} \bmod 10) \bmod 4) \mod 1 = 0$ is always true, and
so is $4\cdot((\pvm{x} \bmod 10)/4)+((\pvm{x} \bmod10) \bmod 4)/1<10$ (remember that `/' refers to integer division
that discards any remainder). This leads to the following new quantifier.
%
%After rewriting the quantifier using these equivalences,
%we get the following.
%\lstinputlisting[style=pvlcomments, firstline=11,lastline=12, numbers=none, aboveskip=1ex, belowskip=1ex]{Chapters/6-Vercors/quantifiers1.pvl}
%This can be further simplified. Note that \pvlinlineb|0<=(((x%10)%4)/1<4| is always true, and so is \pvlinlineb|0<=((x%10)/4)<3|.
%Finally, \pvlinlineb|4*((x%10)/4)+(((x%10)%4)/1)<10| is always true as well (remember that `/` refers to Euclidian division).
%This gives us the following simplified quantifier.
\begin{lstlisting}[style=pvlcomments, firstline=13,lastline=13, numbers=none, aboveskip=1ex, belowskip=1ex]
\forall int x; 0<=x<2*10; A[x]>((x%10)/4);
\end{lstlisting}
\end{example}

Sometimes, the rewriting procedure may result in nonlinear subexpressions.
For example, consider the following quantifier, with positive integers \pvlinlineb|n_1|, \pvlinlineb|n_2|:
\pvlinlineb|\forall int x_1, int x_2; 0<=x_1<n_1 && 0<=x_2<n_2 && x_1%2==0; A[x_1+n_1*x_2]>0|.
This is rewritten to \pvlinlineb|\forall int x; 0<=x<n_1*n_2&&(x%n_1)%2==0; A[x]>0|.
This quantifier contains subexpressions \pvlinlineb|x%n_1| and \pvlinlineb|n_1*n_2|, which are both nonlinear.
It depends on the capabilities of the underlying SMT solver whether this poses a problem.
One could alternatively add additional lemmas to resolve this.

% In order to reason with this, we normally add relevant nonlinear information as verification lemmas. \todo{which lemmas?}
%Typically, you can still work with this, if you add the needed nonlinear information as verification lemmas.

Other verifiers, such as Dafny~\cite{leinoDafnyAutomaticProgram2010} and Viper, often introduce additional structure as a way to circumvent the restriction that arithmetic is not allowed in triggers. With this approach, one could define, for instance, \pvlinlineb|acc(int x, int b, int a) = x*a + b| allowing you to write \pvlinlineb|A[acc(x, b, a)]| as a trigger. One issue is that you must modify the code to be verified to \textit{always} use this indexing. Furthermore, some programs will access arrays in various ways. For example, performing a reduction in a GPU kernel or applying loop-tiling will lead to different access patterns. In those cases, your quantifier will not instantiate anymore. For successful verification, one needs to add complicated verification lemmas relating the different access functions to properly trigger the quantifiers again.

However, the quantifier rewrite method of this section can be \textit{combined} with this additional structure, which is what we have done for the GPU experiments of Section~\ref{sec:experiments}. We use the function \pvlinlineb|int acc1d(int x, int b, int n, int a) =x*a + b| to add additional structure. Therefore, an example trigger used in the complete version of Listing~\ref{lst:swap-cl} has the following structure: \pvlinlineb|xgm[acc1d(\gtid + i*get_global_size(0), x_offset, n, x_inc)]|. The benefit is twofold: (1) we can easily access the array in other ways since less structure is present, and (2) we can add additional nonlinear information to the contract of the \pvlinlineb|acc1d|. This additional information helps in dealing with potential incompleteness caused by nonlinearity.

\paragraph{Implementation in \vercors.}

The rewriting procedure described in this section has been implemented in \vercors. 
%\todo{overall flow picture of the different steps?}
%The user should specify, using trigger syntax, which array index should be attempted to be rewritten. Then, our implementation in VerCors consists of multiple steps:
%\begin{enumerate}
%\item We check whether an index expression fits the linear pattern.
%\item We check whether the conditions are met.
%\item We apply the rewrite.
%\end{enumerate}
%
%The first step is checked by determining whether the array index is a linear combination of the quantified variables. If this is the case, we will check each \textit{permutation} of variables, as we do not know the correct order of the quantified variables a priori. For each $x_i$, we collect all the upper and lower bounds in the domain. Each variable must have a lower bound, and \(x_k\) must have an upper bound.
%
%% \todo{How do we check 1?}
%Next,
To check whether the conditions are met for a quantifier marked as a trigger by the user, \vercors needs to take into account the context of surrounding annotations and statements. Often, equations \ref{eq:pat3-1}--\ref{eq:pat3-3} can
be derived with basic reasoning. Equation~\ref{eq:pat3-4}, however, tends to be harder to determine.
For this reason, we identified Lemma~\ref{lemma:2patterns}, found in  the Appendix, that simplifies checking this condition. We added a straightforward symbolic evaluator to \vercors to verify whether all conditions are satisfied.%\todo{More details about the solver?}.

Finally, even though the rewriting procedure is described for an index expression consisting of a single linear pattern
\pvlinlineb|a_k*x_k + ... + a_1*a_1 + b|,
our implementation contains a generalisation of this, supporting the detection of multiple linear patterns in an expression. By recursively replacing such patterns with single variables using the rewriting procedure, the complete expression is simplified such that it can be used as a trigger.
%!TEX root = main.tex

\section{Unique and Immutable Types for Arrays}\label{sec:types}
Data-level parallel programs tend to work with data stored in arrays, often more than one.
Unfortunately,
as this number of arrays increases, the verification time tends to grow quadratically~\cite{vandenhaakVerifyingRadioTelescope2024},
making programs often
unverifiable. A major cause for this is that \emph{permission quantifiers}, as used in separation logic~\cite{mullerAutomaticVerificationIterated2016}, require numerous internal checks for completeness,
and these quantifiers are essential for verifying parallel programs.
For instance, Listing~\ref{lst:swap-cl} uses permission quantifiers at l.\ref{ln:swap-write-1}--\ref{ln:swap-write-2}, to indicate that a thread executing the
kernel is allowed to write to particular array cells.
%based on the sound method to encode quantified permissions in separation logic by M\"uller et al.~\cite{mullerAutomaticVerificationIterated2016}. Each permission quantifier requires numerous checks internally for completeness, 
The internal checks are needed, since permissions could be combined in case the arrays overlap in memory, and
 program correctness may depend on this. Over time, the heuristics used in deductive verifiers to reason about
 arrays possibly overlapping, and the effects of that happening, have improved, but for complex programs, this problem
 has not been solved.
% As a result, in this case, these checks make the verifier\footnote{For this simple example, some heuristics were applied in newer Silicon versions to make it verifiable. However, the underlying issue still stands for more complex programs: \url{github.com/viperproject/silicon/issues/831}.} run out of time.

However, we observe that in the case of data-level parallel programs, often (1)~the contents of some arrays remain constant throughout the program, and (2)~different arrays typically do not overlap in memory. 
To exploit this, we introduce the \textit{type qualifiers} \pvlinlineb|immutable| and \pvlinlineb|unique| to reduce the number of required verification checks.
Before discussing the type qualifiers, we address why verification times tend to grow rapidly as the number of arrays
increases, at least with the symbolic execution technique for quantified permissions used by \viper~\cite{mullerAutomaticVerificationIterated2016}.
%introduced by M\"uller et al.~\cite{mullerAutomaticVerificationIterated2016}, to further understand why verification times can blow-up when many arrays are present.

A permission quantifier is, internally in \viper, modelled as a \textit{quantified heap chunk}, i.e., a function that takes a memory reference $\pvm{r}.\pvm{f}$, with $\pvm{r}$ a quantified (reference-typed) receiver and $\pvm{f}$ a field,
 and returns a tuple $(v(\pvm{r}),p(\pvm{r}))$, consisting of a symbolic value $v(\pvm{r})$
%, also called a \textit{field value function}, 
and a (symbolic) permission $p(\pvm{r}) \in \langle0,1]$. For example, before processing quantifier \pvlinlineb|\forall int i; 0<=i<9 ==> Perm(xs[i], 1\2)|, a new chunk is produced mapping each of the first nine elements of \pvlinline|xs|
to a symbolic field value and a \pvlinline|1\2| permission.
Subsequent quantifiers for the same array(s) produce additional chunks, as opposed to updated existing ones.
Whenever an assertion is processed, a \emph{snapshot} is created, which contains the chunks that are relevant for
the assertion, i.e., that contain the symbolic values referred to in the assertion~\cite{schwerhoffAdvancingAutomatedPermissionBased2016}. To avoid inconsistencies, \viper needs to periodically compare snapshots to check if they are the same. This is because a function call to a side-effect-free function is considered equal if all the arguments are the same, including heap-dependent arguments. This involves
comparing all the chunks of one snapshot with all the chunks of the other snapshot, causing the quadratic
execution time.

\paragraph*{Immutable qualifier.}
With the \pvlinlineb|immutable| qualifier for (a pointer to) an array, for instance in \pvlinlineb|immutable int* xs|,
we propose to indicate that the pointer refers to data that does not change throughout execution of the program.
For such pointers, in \vercors, the data can be encoded using an immutable \emph{sequence}, as opposed to a
mutable pointer block. Elements of a sequence are not stored on the heap.

% in C  to a variable's type, indicates that the variable cannot be modified during its lifetime after initial assignment (and thus is read-only). You can also apply the \pvlinline|const| modifier to a pointer, such as \pvlinline|const int* xs|, meaning you cannot assign to \pvlinline|xs[i]| for any index i. 

%\vercors encodes normal pointers as a \textit{pointer block} and an \textit{offset}. For instance, declaring a pointer like \pvlinline|int* xs = (int*) malloc(sizeof(int)*n)| creates a pointer with an offset \pvlinline|0| and an underlying mutable pointer block of size \pvlinline|n|. However, 
%an \pvlinline|immutable| pointer allows us to model the underlying pointer block as an immutable \sequence. A sequence is a \vercors built-in type that resembles an array, but its elements are immutable and not stored on the heap. Listing~\ref{lst:const-pointer} in the Appendix shows the Axiomatic Data Type (ADT) encoding of the immutable pointer.

With this new qualifier, functions can be defined with immutable pointer parameters.
However, when called, the given pointer argument may not be immutable. To allow this,
the mutable pointer needs to be coerced to an immutable one. Once this has happened,
the pointer cannot be changed anymore. This is achieved by requiring that the mutable pointer
releases a positive amount of permission for each of its elements when coerced.

\paragraph*{Unique type qualifier.}
We propose the \pvlinlineb|unique#-i-#| type qualifier to distinguish non-overlapping data, allowing for
a more efficient encoding in \viper. Here, \pvlinlineb|i| is a \textit{uniqueness number}: only
pointers of the same type \emph{and} that have the same uniqueness number can potentially alias.
Unique type qualifiers are applicable for all heap locations, such as arrays, pointers, members of structs,
and fields of classes.
In Listing~\ref{lst:unique}, we show several examples of how this qualifier can be used,
including some restrictions.
%Correctness of our approach is guaranteed by the type checker.
For instance, \pvlinline|xs1|  (l.\ref{lst:unique-xs1}) and \pvlinline|xs3| (l.\ref{lst:unique-xs3}) have the same type, but different uniqueness numbers, making the assignment at l.\ref{lst:unique-assignx1x3} illegal.
%If line~\ref{lst:unique-a} were present without line~\ref{lst:querry-a}, the unique type modifier would have no effect. %However, since
At l.\ref{lst:querry-a}, the address of \pvlinline|a| is queried, resulting in the variable being stored on the heap as an integer with uniqueness number 1 (l.\ref{lst:unique-a}).

%\lstinputlisting[style=pvl2, linerange={3-14,17-21,28-33},float=t,caption={\label{lst:unique}Examples of unique type modifier used on pointers and structs.}]{Chapters/6-Vercors/unique-types.c}
\begin{lstlisting}[style=pvl2, float=t,caption={\label{lst:unique}Examples of unique type modifier used on pointers and structs.}]
/*@unique#-1-#@*/ int*  xs1; // The pointer points to integers with uniqueness number 1 |\label{lst:unique-xs1}|
int /*@unique#-1-#@*/*  xs2; // The same type as above |\label{lst:unique-xs2}|
xs1 = xs2; // Allowed: both have the same uniqueness number for their integers |\label{lst:unique-assignx1x2}|
int /*@unique#-2-#@*/*  xs3; // This involves integers with uniqueness number 2 |\label{lst:unique-xs3}|
//~xs1 = xs3; // Not allowed: the pointers have different uniqueness numbers |\label{lst:unique-assignx1x3}|
/*@unique#-1-#@*/ int a; // Uniqueness not relevant, as integers are not stored on heap|\label{lst:unique-a}|
xs1 = &a; // Address of a is tracked as heap variable, hence uniqueness number matters |\label{lst:querry-a}|
int* /*@unique#-1-#@*/ * zs; // The zs pointer has uniqueness number 1 |\label{lst:unique-zs}|
xs1 = *zs; // Not allowed: *zs values do not have a uniqueness number |\label{lst:unique-assignx1z}|
struct v { int n; int* xs; /*@unique#-2-#@*/ int* ys; }; // Uniqueness used in a struct |\label{lst:unique-struct}|
struct v a; // A struct stored in variable a with uniqueness numbers as defined for v |\label{lst:unique-struct-start}|
/*@unique_field#-n,1-#@*/  struct v b; // Member n has uniqueness number 1
/*@unique_pointer_field#-xs,1-#@*/  struct v c; // Type of xs member is unique<1> int *
//~a.xs = xs1; // Not allowed: a.xs* values do not have a uniqueness number
c.xs = xs1; // Allowed: the uniqueness numbers match
//~xs1 = &(b.n); // Allowed: the uniqueness numbers match |\label{lst:unique-struct-end}|
void sort(int* xs, int len){...} // A previously defined sorting function |\label{lst:sort}|
int f(/*@unique#-1-#@*/int* xs, int n){
  sort(xs, n);//~|\label{lst:modulo-uniqueness}|
}
\end{lstlisting}
%\lstinputlisting[style=pvl, linerange={36-39,42-47},float=t,caption={Adding unique types to existing struct members.\label{lst:unique-struct}}]{Chapters/6-Vercors/unique-types.c}

%, and variables with addresses queried in \C.
%Regarding type constraints, in addition to the usual constraints, we require them to share the same uniqueness type number.

Like \pvlinlineb|const| in \C, the unique type qualifier associates to the left, and pointer markers (\pvlinlineb|*|) denote uniqueness at different levels.
For example,
%the unique type of \pvlinline|ys| on line~\ref{lst:unique-ys} has no effect, since \pvlinline|ys| is not stored on the heap. However,
at l.\ref{lst:unique-zs}, the content of \pvlinlineb|zs|, of integer pointer type, has uniqueness number 1. Yet,
the inner pointer of \pvlinlineb|zs| has no uniqueness number, so the assignment at l.\ref{lst:unique-assignx1z} is not allowed.

Unique types can also be applied to struct members, see l.\ref{lst:unique-struct}. Sometimes,
it may be known that
different \pvlinlineb|struct| instances do not overlap their internal data. For this purpose, we propose the
type qualifier annotations \pvlinlineb|unique_field| and \pvlinlineb|unique_pointer_field|.
%These qualifiers
%can be used to indicate that (pointer) members of structs do not overlap, even if contained within the same struct.
These features are demonstrated at l.\ref{lst:unique-struct-start}--\ref{lst:unique-struct-end}.

%\subsubsection{Encoding}
Heap locations in \viper are encoded using \textit{fields}, with all memory locations of the same type using the same field. When using the unique type qualifiers, types with different uniqueness numbers use different fields.
%For instance, an array of type \pvlinlineb|unique#-1-# bool* xs| stores its values in the field \pvlinlineb|bool1: Bool|.
When comparing snapshots, as mentioned earlier, only chunks that refer to the same field are compared.
Therefore, by separating non-overlapping arrays into different fields, significant verification time improvements can
be achieved. For more on the practical impact of this, see Section~\ref{sec:experiments}.

% allowing access to its value via \pvlinline|loc(xs, i).bool1| for the array \pvlinline|xs| and index \pvlinline|i|.

%\subsubsection{Correctness of encoding}
%Regarding correctness of our approach, ensure correct encoding, the program must never place a heap location of type \pvlinline|unique<A> t| in the heap location of \pvlinline|unique<B> t| if \pvlinline|A!=B|. Fortunately, type checking resolves this issue. A heap location changes only via assignment or function calls, and in both cases, if the type checker does not allow coercion between \pvlinline|unique<A> t| and \pvlinline|unique<B> t| when \pvlinline|A!=B|, a value cannot be placed in the incorrect heap location. Thus, the type checker disallows line \ref{lst:unique-assignx1x3} of Listing~\ref{lst:unique}.
%However, we permit coercion between (unique) types if the type is not stored on the heap. Thus, \pvlinline|unique<1> int| can be coerced to \pvlinline|int|; otherwise, operations like \pvlinline|a[0] = 1| for an array \pvlinline|unique<1> int* a| would be impossible.

%\subsubsection{Correct coercions for function calls}
Finally, previously defined functions, as, for instance, provided
in libraries, should be usable in combination with unique type qualifiers.
%For this reason,
%we also 
%allow some type coercions between unique types in function calls, otherwise, unique types work poorly with previously defined functions, such as library functions. For example, consider Listing~\ref{lst:sort-coercion}.
%\lstinputlisting[style=pvl, float=t, linerange={50-54,57-62},caption={Calling a sort function for pointers of different uniqueness.\label{lst:sort-coercion}}]{Chapters/6-Vercors/unique-types.c}
For instance, consider l.\ref{lst:modulo-uniqueness}, where the integer array \pvlinlineb|xs| with uniqueness number 1
is given as parameter to a predefined sorting function (l.\ref{lst:sort}), in which uniqueness numbers are not given.
This will not type check, as \pvlinline|sort| accepts only non-unique arrays.
%However, the unique modifier indicates only that distinct arrays do not overlap; with a single array as a parameter, there is nothing to overlap with. Moreover, there is no fundamental difference among \pvlinline|unique<1> int*|, \pvlinline|unique<2> int*|, and \pvlinline|int*|. Therefore, 
However, such function calls should be supported regardless of the exact uniqueness number, as long as uniqueness between arrays is respected
%, provided uniqueness between arrays is respected. We achieve this with type coercions in \vercors during function calls and check if the type coercions of a function call are \textit{consistent}. If they are consistent, we automatically make a function copy with the correct type signature.
We consider a function call \emph{consistent} if the type signature of a call partitions the parameters in the same
way as the called function. For example, a function \pvlinlineb|f(unique#-0-# int* x, unique#-0-# int* y)| can be called
with both arguments of type \pvlinlineb|unique<1> int*|, but not if one has type \pvlinlineb|unique<0> int*| and the
other has type \pvlinlineb|unique<1> int*|.

\iffalse
\begin{table}[!t]
\centering
\input{results/blas_level1.tex}
\input{results/blas_level2.tex}
\caption{\label{tab:blas-results}Verification results for the CLBlast GPU kernels for level 1 and 2, with reported backend verification time (T) in seconds.
% Level 1 in BLAS covers vector operations only, while level 2 covers matrix–vector operations. 
\textbf{Base} is the baseline configuration without any extra features. It is compared against versions with unique type qualifiers (\textbf{Unique}), immutable arrays (\textbf{Immutable}), extracted kernel bodies (\textbf{Extract}), and a configuration enabling all features simultaneously (\textbf{All}).
The $\dagger$ indicates we (try to) prove functional correctness in addition to memory safety.
% A name marked with a $\dagger$ indicates that \textsc{VerCors} attempts to prove both memory safety and functional correctness of the kernel; otherwise, it only attempts to prove memory safety.
% $T_t$ denotes the average total verification time of VerCors, and $T_v$ the average time spent in the VerCors backend. Speedup$_v$ is the speedup of backend time relative to the baseline.
$\#_{imm}$ and $\#_p$ is the number of immutable and normal pointer arrays, respectively. $\#_A$ the is the total number of annotations.}
\end{table}
\else
\begin{figure}[t!]
\centering
\subfloat[\label{fig:blas-barplot-level1}Level 1]{\includegraphics[width=0.95\textwidth]{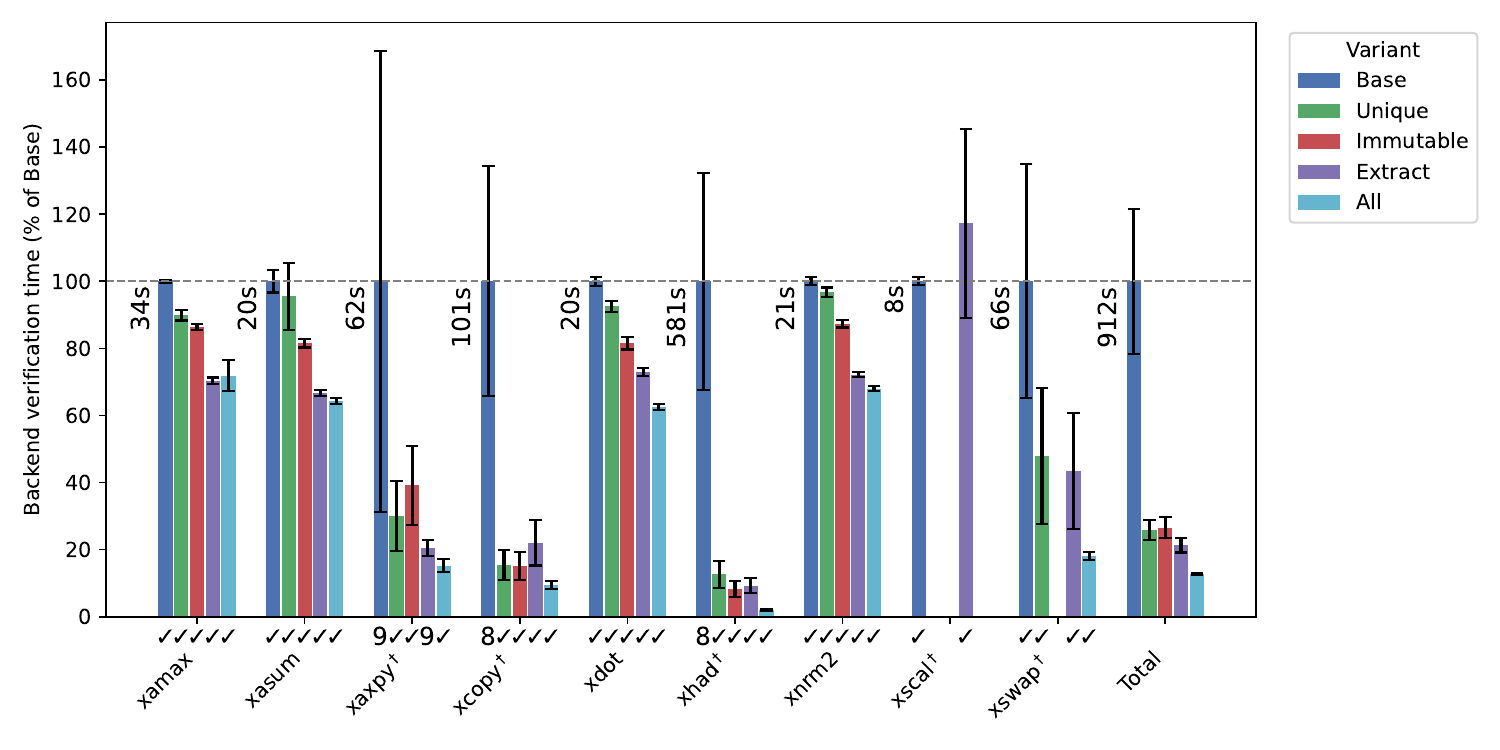}}\\
\subfloat[\label{fig:blas-barplot-level2}Level 2]{\includegraphics[width=0.95\textwidth]{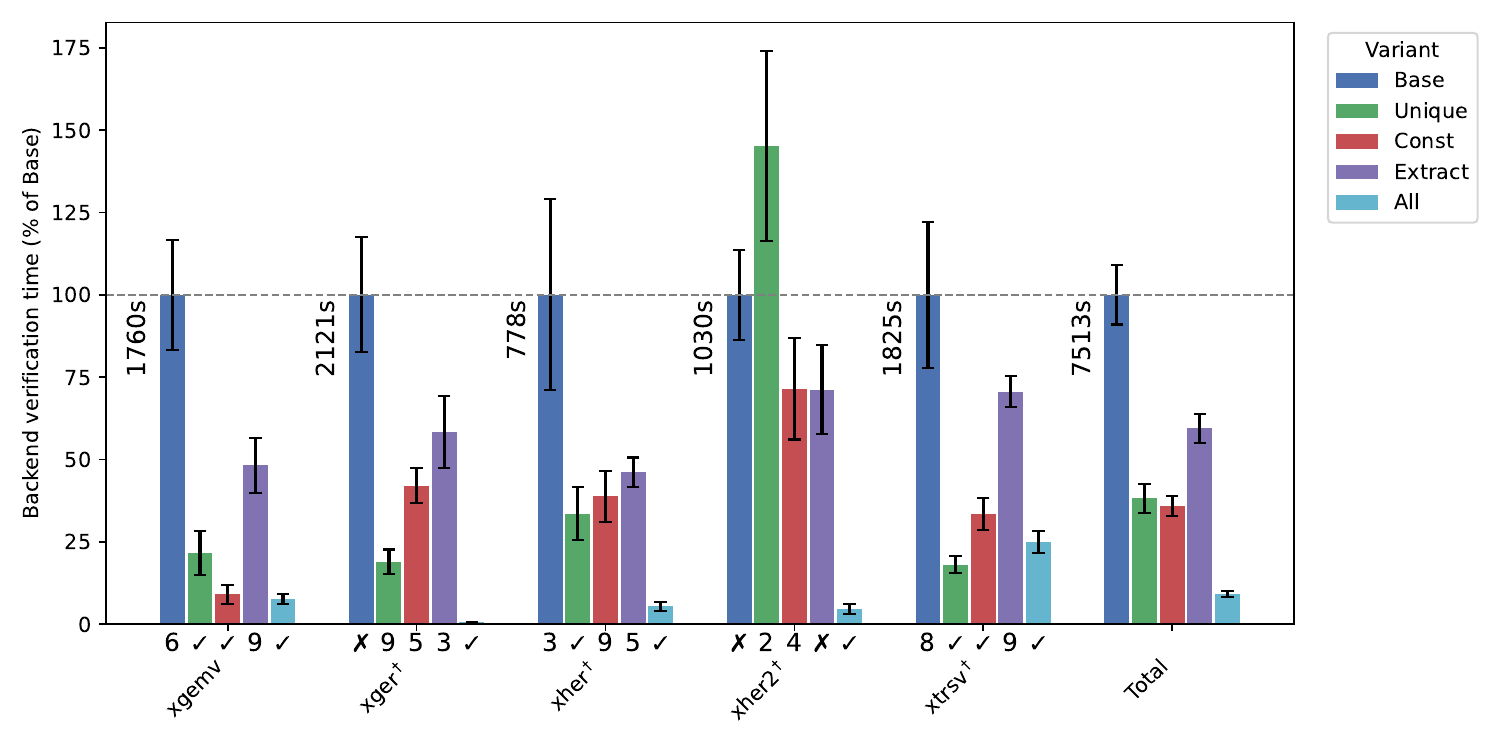}}
\caption{\label{fig:blas-results}Verification of CLBlast kernels, normalized so Base=100\%. 
%\textbf{Base} is the baseline configuration without any extra features.
\textbf{Base} is compared against versions with unique type qualifiers (\textbf{Unique}), immutable arrays (\textbf{Immutable}), extracted kernel bodies (\textbf{Extract}), and \textbf{All} which enables all these features. %Bars show mean with standard error of the mean.
Below the bars: number of successful runs (\checkmark = 10/10, \ding{55} 0/10).
%\checkmark = 10/10 successful, \ding{55}  = 0/10 successful, else number of successful runs out of 10.
Next to base bars: average time of the 10 base runs in seconds. The $\dagger$ indicates we (try to) prove functional correctness in addition to memory safety.}
\end{figure}
\fi

\section{Experiments}\label{sec:experiments}
% \subsection{Experiments}\todo{replace by new experiments section?}
In this section, we evaluate
%\footnote{Found in the accompanying artefact.}
the effectiveness of the techniques proposed in Sections~\ref{sec:quantifiers} and \ref{sec:types} to verify
data-level parallel programs. Regarding the rewrite procedure of Section~\ref{sec:quantifiers}, the conclusion is clear:
none of the experiments we report on in this section were verifiable without the use of this procedure. Using the type qualifiers \pvlinlineb|unique| and \pvlinlineb|immutable|
%, along with their encodings in \viper as described in Section~\ref{sec:types},
further improves the verification of data-level parallel programs that contain multiple arrays.
%Additionally, the
%quantifier rewrite method of Section~\ref{sec:quantifiers} is also evaluated, as all the experiments shown here rely on it. 

We chose three representative sets of experiments that heavily rely on quantifiers and flattened multi-dimensional arrays. The first set contains the GPU kernels from the CLBlast library~\cite{nugterenCLBlastTunedOpenCL2018}, consisting of OpenCL kernels that implement the Basic Linear Algebra Subprograms (BLAS). These experiments are conducted to 
evaluate whether our proposed techniques are applicable to typical GPU kernels.
The next two sets of experiments were chosen 
from~\cite{vandenhaakHaliVerDeductiveVerification2024,vandenhaakVerifyingRadioTelescope2024},
among which is an implementation of the Padre algorithm, as part of the software for a Radio Telescope Pipeline~\cite{vandenhaakVerifyingRadioTelescope2024}.
In these papers, the authors ran into limitations of the underlying verifiers due to the many arrays and quantifiers present. These works contain optimised parallel CPU programs generated from the Domain Specific Language Halide~\cite{ragan-kelleyHalideDecouplingAlgorithms2017}. The HaliVer~\cite{vandenhaakHaliVerDeductiveVerification2024} tool adds verification annotations to these programs such that they can be verified by \vercors.

\paragraph{Set-up.} We used a machine with an 11th Gen Intel Core i7-11800H @ 2.30GHz and 32GB of RAM running Ubuntu 23.04. We ran each experiment ten times and reported the average of these. The \vercors version used is included in the accompanying artefact. We ran \vercors with the options \verb|--silicon-quiet| \verb|--dev-time-backend| \verb|--dev-total-timeout=3600| \verb|--dev-assert-timeout 60| \verb|--target x86_64-linux-gnu|, and used the Silicon verifier of \viper, which is based on symbolic execution. Additionally, the experiments from~\cite{vandenhaakHaliVerDeductiveVerification2024,vandenhaakVerifyingRadioTelescope2024} were also run with \verb|--no-infer-heap-context-into-frame| as these options were also used in the original experiments.

\begin{table}[!t]
\caption{Reverification of the experiments presented in~\cite{vandenhaakHaliVerDeductiveVerification2024}.
% The left table considers functional correctness, whilst the right table only considers memory safety.
% We compare using \pvlinline|unique| and \pvlinline|const| type qualifiers (\textbf{Qualifiers}) with not using them (\textbf{Base}).
% The `Base` experiments show different times compared to the originals in~\cite{vandenhaakHaliVerDeductiveVerification2024} since we reevaluated them using the same version as for the `Qualifiers` experiments for a fair comparison.
\label{tab:haliver}
}
\centering
\label{tab:chp6-haliver-mem}\small
 \centering
\begin{subtable}[t]{0.475\textwidth}
{
\resizebox{1\textwidth}{!}{  
\begin{tabular}[t]{lll|rrr|rrr|r}
\hline
 & & & \multicolumn{3}{c|}{\textbf{Base}} & \multicolumn{3}{c|}{\textbf{Qualifiers}} & \\
\textbf{Name} & \textbf{V} & \textbf{Result} & \textbf{\#} & \textbf{T} & $\mathbf{\sigma}$ & \textbf{\#} & \textbf{T} & $\mathbf{\sigma}$ & \textbf{Speedup} \\
\hline
blur\ & 0& \checkmark&  10 &  13 & 1&  10&  9 & 1& \cellcolor{ForestGreen!25} 1.4 \\
\hline
 & 1& \checkmark&  10 &  14 & 1&  10&  10 & 1& \cellcolor{ForestGreen!25} 1.4 \\
\hline
 & 2& \checkmark&  10 &  18 & 1&  10&  13 & 1& \cellcolor{ForestGreen!25} 1.4 \\
\hline
 & 3& \checkmark&  10 &  20 & 1&  10&  20 & 7&  1.0 \\
\hline
hist\ & 0& \checkmark&  10 &  25 & 1&  10&  26 & 7&  1.0 \\
\hline
 & 1& \checkmark&  10 &  30 & 1&  10&  17 & 1& \cellcolor{ForestGreen!25} 1.7 \\
\hline
 & 2& \checkmark&  10 &  42 & 1&  10&  23 & 1& \cellcolor{ForestGreen!25} 1.8 \\
\hline
 & 3& \checkmark&  10 &  79 & 3&  10&  32 & 2& \cellcolor{ForestGreen!25} 2.5 \\
\hline
conv\_layer\ & 0& \checkmark&  10 &  153 & 1&  10&  53 & 1& \cellcolor{ForestGreen!25} 2.9 \\
\hline
 & 1& \checkmark&  10 &  171 & 1&  10&  72 & 6& \cellcolor{ForestGreen!25} 2.4 \\
\hline
 & 2& \checkmark&  10 &  246 & 3&  10&  71 & 1& \cellcolor{ForestGreen!25} 3.5 \\
\hline
 & 3& \checkmark&  10 &  227 & 2&  10&  68 & 1& \cellcolor{ForestGreen!25} 3.3 \\
\hline
gemm\ & 0& \checkmark&  10 &  61 & 1&  10&  22 & 1& \cellcolor{ForestGreen!25} 2.8 \\
\hline
 & 1& \checkmark&  10 &  133 & 1&  10&  45 & 4& \cellcolor{ForestGreen!25} 2.9 \\
\hline
 & 2& \checkmark&  10 &  250 & 11&  10&  131 & 18& \cellcolor{ForestGreen!25} 1.9 \\
\hline
 & 3& Error& 10 &  & & 10 &  &  \\
\hline
auto\_viz\ & 0& \checkmark&  10 &  19 & 1&  9&  21 & 2& \cellcolor{BrickRed!25} 0.9 \\
& & \ding{55}& 0 &  & & 1 & 15 & - \\
\hline
 & 1& \checkmark&  8 &  63 & 1&  10&  52 & 1& \cellcolor{ForestGreen!25} 2.1 \\
& & \ding{55}& 2 & 290 & 1& 0 &  &  \\
\hline
 & 2& \checkmark&  10 &  67 & 1&  10&  55 & 1& \cellcolor{ForestGreen!25} 1.2 \\
\hline
 & 3& \checkmark&  10 &  49 & 1&  10&  36 & 1& \cellcolor{ForestGreen!25} 1.4 \\
\hline
\hline
Total & & & & 1722 & 33  & & 774 & 22  & \cellcolor{ForestGreen!25} 2.2  \\
\hline
\end{tabular}
}
}
\caption{Functional correctness \& Memory safety}
\end{subtable}
\begin{subtable}[t]{0.51\textwidth}
{
\resizebox{1\textwidth}{!}{  
\begin{tabular}[t]{lll|rrr|rrr|r}
\hline
 & & & \multicolumn{3}{c|}{\textbf{Base}} & \multicolumn{3}{c|}{\textbf{Qualifiers}} & \\
\textbf{Name} & \textbf{V} & \textbf{Result} & \textbf{\#} & \textbf{T} & $\mathbf{\sigma}$ & \textbf{\#} & \textbf{T} & $\mathbf{\sigma}$ & \textbf{Speedup} \\
\hline
blur\ & 0& \checkmark&  10 &  13 & 1&  10&  10 & 1& \cellcolor{ForestGreen!25} 1.3 \\
\hline
 & 1& \checkmark&  10 &  13 & 1&  10&  10 & 1& \cellcolor{ForestGreen!25} 1.3 \\
\hline
 & 2& \checkmark&  10 &  17 & 1&  10&  14 & 3& \cellcolor{ForestGreen!25} 1.2 \\
\hline
 & 3& \checkmark&  10 &  14 & 1&  10&  11 & 1& \cellcolor{ForestGreen!25} 1.4 \\
\hline
hist\ & 0& \checkmark&  10 &  22 & 1&  10&  15 & 1& \cellcolor{ForestGreen!25} 1.5 \\
\hline
 & 1& \checkmark&  10 &  27 & 1&  10&  20 & 2& \cellcolor{ForestGreen!25} 1.4 \\
\hline
 & 2& \checkmark&  10 &  36 & 1&  10&  32 & 10& \cellcolor{ForestGreen!25} 1.1 \\
\hline
 & 3& \checkmark&  10 &  46 & 1&  10&  29 & 2& \cellcolor{ForestGreen!25} 1.6 \\
\hline
conv\_layer\ & 0& \checkmark&  10 &  144 & 2&  10&  60 & 6& \cellcolor{ForestGreen!25} 2.4 \\
\hline
 & 1& \checkmark&  10 &  162 & 1&  10&  65 & 1& \cellcolor{ForestGreen!25} 2.5 \\
\hline
 & 2& \checkmark&  10 &  220 & 3&  10&  74 & 5& \cellcolor{ForestGreen!25} 3.0 \\
\hline
 & 3& \checkmark&  10 &  201 & 2&  10&  72 & 5& \cellcolor{ForestGreen!25} 2.8 \\
\hline
gemm\ & 0& \checkmark&  10 &  56 & 1&  10&  21 & 1& \cellcolor{ForestGreen!25} 2.6 \\
\hline
 & 1& \checkmark&  10 &  110 & 1&  10&  35 & 1& \cellcolor{ForestGreen!25} 3.2 \\
\hline
 & 2& \checkmark&  10 &  152 & 8&  10&  53 & 6& \cellcolor{ForestGreen!25} 2.9 \\
\hline
 & 3& Error& 10 &  & & 10 &  &  \\
\hline
auto\_viz\ & 0& \checkmark&  10 &  18 & 1&  10&  21 & 4& \cellcolor{BrickRed!25} 0.9 \\
\hline
 & 1& \checkmark&  9 &  34 & 1&  10&  45 & 10& \cellcolor{ForestGreen!25} 1.3 \\
& & \ding{55}& 1 & 266 & -& 0 &  &  \\
\hline
 & 2& \checkmark&  10 &  36 & 1&  10&  39 & 10& \cellcolor{BrickRed!25} 0.9 \\
\hline
 & 3& \checkmark&  10 &  34 & 1&  10&  28 & 1& \cellcolor{ForestGreen!25} 1.2 \\
\hline
\multicolumn{2}{l}{bilateral\_grid}& \checkmark&  10 &  48 & 1&  10&  39 & 1& \cellcolor{ForestGreen!25} 1.2 \\
\hline
\multicolumn{2}{l}{camera\_pipe}& \checkmark&  10 &  205 & 7&  10&  122 & 14& \cellcolor{ForestGreen!25} 1.7 \\
\hline
\multicolumn{2}{l}{depthwise\_}& \checkmark&  10 &  201 & 2&  10&  143 & 1& \cellcolor{ForestGreen!25} 1.4 \\
separable\_conv & & & & & & & & \\
\hline
\hline
Total & & & & 1831 & 26  & & 956 & 26  & \cellcolor{ForestGreen!25} 1.9  \\
\hline
\end{tabular}
}
}
\caption{Memory safety only}
\end{subtable}
\end{table}

\begin{table}[t]
\caption{Reverification of the experiments presented in~\cite{vandenhaakVerifyingRadioTelescope2024}.
% We compare using \pvlinline|unique| and \pvlinline|const| type qualifiers (\textbf{Qualifiers}) with not using them (\textbf{Base}).
% The `Base` experiments show different times compared to the originals in~\cite{vandenhaakHaliVerDeductiveVerification2024} since we reevaluated them using the same version as for the `Qualifiers`. We tested 4 different algorithms, each with concrete bounds (CB) or non-concrete bounds (NCB).
\label{tab:chp6-padre}
}

\centering
 \small
 \centering
{
\newcommand{\widthPadre}{1}
\begin{subtable}[t]{0.49\textwidth}
\resizebox{\widthPadre\textwidth}{!}{
\begin{tabular}[t]{ll|rrr|rrr|r}
\hline
 & & \multicolumn{3}{c|}{\textbf{Base}} & \multicolumn{3}{c}{\textbf{Qualifiers}} & \\
\textbf{Version} & \textbf{Result} & \textbf{\#} & \textbf{T} & $\mathbf{\sigma}$ & \textbf{\#} & \textbf{T} & $\mathbf{\sigma}$ & \textbf{Speedup} \\
\hline
CB& \checkmark&  10 &  56 & 1&  10&  48 & 1& \cellcolor{ForestGreen!25} 1.2 \\
NCB& \checkmark&  10 &  56 & 1&  10&  50 & 1& \cellcolor{ForestGreen!25} 1.1 \\
\end{tabular}
}
\caption{\label{tab:StepHalide}\texttt{\texttt{step}}}
\end{subtable}
\hfill
\begin{subtable}[t]{0.49\textwidth}
\resizebox{\widthPadre\textwidth}{!}{
\begin{tabular}[t]{ll|rrr|rrr|r}
\hline
 & & \multicolumn{3}{c|}{\textbf{Base}} & \multicolumn{3}{c}{\textbf{Qualifiers}} & \\
\textbf{Version} & \textbf{Result} & \textbf{\#} & \textbf{T} & $\mathbf{\sigma}$ & \textbf{\#} & \textbf{T} & $\mathbf{\sigma}$ & \textbf{Speedup} \\
\hline
CB& \checkmark&  10 &  466 & 9&  10&  142 & 4& \cellcolor{ForestGreen!25} 3.3 \\
NCB& \checkmark& \cellcolor{BrickRed!25} 0 &   & & \cellcolor{ForestGreen!25} 1&  901 & -& \cellcolor{ForestGreen!25} 1.3 \\
& \ding{55}& 10 & 1028 & 58& 9 & 780 & 89 \\
\end{tabular}
}
\caption{\label{tab:SubDirectionHalide}\texttt{\texttt{sub\_direction}}}
\end{subtable}
\\
\begin{subtable}[t]{0.49\textwidth}
\resizebox{\widthPadre\textwidth}{!}{
\begin{tabular}[t]{ll|rrr|rrr|r}
\hline
 & & \multicolumn{3}{c|}{\textbf{Base}} & \multicolumn{3}{c}{\textbf{Qualifiers}} & \\
\textbf{Version} & \textbf{Result} & \textbf{\#} & \textbf{T} & $\mathbf{\sigma}$ & \textbf{\#} & \textbf{T} & $\mathbf{\sigma}$ & \textbf{Speedup} \\
\hline
CB& \checkmark& \cellcolor{BrickRed!25} 0 &   & & \cellcolor{ForestGreen!25} 10&  961 & 14& \cellcolor{ForestGreen!25} 3.4 \\
& \ding{55}& 2 & 2086 & 1404& 0 &  &  \\
& T.O.& 8 & - & -& 0 &  &  \\
NCB& \ding{55}& 9 & 1139 & 131& 6 & 1211 & 468 \\
& T.O.& 1 & - & -& 4 & - & - \\
\end{tabular}
}
\caption{\label{tab:SolveDirectionHalide}\texttt{\texttt{solve\_direction}}}
\end{subtable}
\hfill
\begin{subtable}[t]{0.49\textwidth}
\resizebox{\widthPadre\textwidth}{!}{
\begin{tabular}[t]{ll|rrr|rrr|r}
\hline
 & & \multicolumn{3}{c|}{\textbf{Base}} & \multicolumn{3}{c}{\textbf{Qualifiers}} & \\
\textbf{Version} & \textbf{Result} & \textbf{\#} & \textbf{T} & $\mathbf{\sigma}$ & \textbf{\#} & \textbf{T} & $\mathbf{\sigma}$ & \textbf{Speedup} \\
\hline
CB& \checkmark& \cellcolor{BrickRed!25} 0 &   & & \cellcolor{ForestGreen!25} 10&  2799 & 38& \cellcolor{ForestGreen!25} 1.2 \\
& \ding{55}& 6 & 3146 & 101& 0 &  &  \\
& T.O.& 4 & - & -& 0 &  &  \\
NCB& \ding{55}& 1 & 3589 & -& 0 &  &  \\
& T.O.& 9 & - & -& 10 & - & - \\
\end{tabular}
}
\caption{\label{tab:PerformIterationHalide}\texttt{\texttt{perform\_iteration}}}
\end{subtable}
\\
}
\end{table}

\paragraph{GPU body extraction.}\label{sec:extract}
Verifying monolithic functions can be more costly than verifying code that is functionally the same but is divided into several functions. The slowdown mainly occurs because the underlying SMT solver has \textit{too much} information. Additionally, as mentioned in Section~\ref{sec:types}, too many program statements increase the number of quantified chunks, which harms performance.  Strategies for decomposing large programs into smaller functions for verification have been previously explored in Dafny~\cite{leinoDafnyAutomaticProgram2010} and Gobra~\cite{wolfGobraModularSpecification2021}. For GPU kernels, we observed major improvements when the kernel body is \textit{extracted} into a separate function with its own contract and verified in isolation. Since these contracts are almost identical to standard \vercors GPU-kernel contracts, we added an option to perform this extraction automatically via the \pvlinline|extract_body| annotation for GPU kernels. We also evaluated this feature in our GPU experiments.

\paragraph{Results.} The results are shown in Figure~\ref{fig:blas-results} and the Tables 
%Tables~\ref{tab:blas-results}
~\ref{tab:haliver} and \ref{tab:chp6-padre}. 
All times are in seconds, and we set a time out of 1 hour. We only report the backend verification time ($T$) and the standard deviation of the mean times ($\sigma_{mean}$).
%, as the differences in time come from different encodings in Viper, and the time taken to parse and translate the input is not significantly influenced. 
For the tables, we compare using \pvlinlineb|unique| and \pvlinlineb|immutable| type qualifiers (\textbf{Qualifiers}) with not using them (\textbf{Base}). In all configurations, however, the rewriting procedure was applied, as this was essential
to make the experiments succeed.
The verification result can be successful ($\checkmark$), fail (\ding{55}), or time out (T.O.). With $\#$, we indicate how many times an experiment yielded a certain result, summing up to 10 for each experiment.  Speed-up is calculated by dividing $T$ of the \textbf{Base} version by $T$ of the specific version for all results ($\checkmark$, \ding{55} and T.O).

\paragraph{Evaluation of Section~\ref{sec:types}.}

The CLBlast experiments 
% (Table~\ref{tab:blas-results})
(Figure~\ref{fig:blas-results})
cover the level 1 and level 2 subroutines of BLAS. Level 1 covers vector operations only, while level 2 covers matrix–vector operations, which are more complicated. Kernels with no immutable pointers (xswap, xscal) are not evaluated for \pvlinlineb|immutable|, and kernels with only one pointer (xscal) are not evaluated for \pvlinlineb|unique|. In both these cases, those versions would be the same as the \textbf{Base} version. We fixed the thread block sizes of these experiments to a concrete value, as otherwise verification of these kernels could run into incompleteness due to nonlinearity.
% As mentioned before, nonlinearity can lead to the incompleteness of verification. Therefore, we had to make the thread block sizes concrete in most kernels. Most kernels are implemented for different types (floats, doubles, imaginary numbers) or different matrices (hermitian, symmetric, etc.) using macros. We only verified one specific version of each kernel since all the different options lead to the same structure, and we would expect almost the same verification results there.

These experiments show that applying type qualifiers is always a good idea and can be up to 10 times faster. Similarly, extracting the body of a kernel is always a good idea and speeds up verification (for `xscal', the difference is not statistically relevant). However, the combination of the techniques shows the most impressive results, being \textbf{up to 150 times faster} for `xger'. Especially for the more complex level 2 kernels, this leads to the results being less brittle (all have ten
out of ten successful verifications) or even being able to verify `xger' and `xher2' at all.  In total, the combination of features (All) is 8.5 times faster for the total time of level 1 and 10.8 times faster for level 2.

For the experiments taken from~\cite{vandenhaakHaliVerDeductiveVerification2024} (Table~\ref{tab:haliver}), \textbf{V} indicates that each program has up to four different Halide \textit{schedules} that explore different optimisations that influence parallelisation and the ordering of computations.
%For these experiments,
\textbf{Qualifiers} combines the use of the \pvlinlineb|immutable| and \pvlinlineb|unique| encodings. %from Section~\ref{sec:types}.

These experimental results demonstrate that using the \pvlinlineb|immutable| and \pvlinlineb|unique| qualifiers mostly has a positive influence. In total, the verification time is reduced by a factor of more than 1.9. Again, this is a speedup achieved 
after having made the benchmarks verifiable to begin with, by using the rewriting procedure.

For the experiments taken from~\cite{vandenhaakVerifyingRadioTelescope2024} (Table~\ref{tab:chp6-padre}), we make a distinction between versions with concrete bounds (CB) and non-concrete bounds (NCB). The NCB versions contain nonlinear arithmetic, which can lead to incompleteness. The authors of~\cite{vandenhaakVerifyingRadioTelescope2024}
tried to address this by adding verification lemmas about nonlinear arithmetic, which only led to successful verification in some cases. In the original experiments, the complete algorithm (\texttt{perform\_iteration}) was not included as it could not be verified. We, however, have been able to verify
the CB version.

The experimental results demonstrate that when successful verification is possible, use of the encodings of Section~\ref{sec:types} always speeds up the verification. In addition, more importantly, it allows three cases
to be verified that cannot be verified without using the type qualifier encodings.

\paragraph{Evaluation of Section~\ref{sec:quantifiers}.} All experiments above rely on the rewrite method of Section~\ref{sec:quantifiers}. Without it, the benchmarks cannot be verified as they are. A similar quantifier rewriter,
in a rudimentary form, had been implemented in \vercors before, but that approach had not been formally proven
correct, nor described in the literature. However, for the experiments in~\cite{vandenhaakHaliVerDeductiveVerification2024,vandenhaakVerifyingRadioTelescope2024}, the authors
depended on this rewriter for their results. The current paper proposes a more generally applicable rewriter,
involving better symbolic checks to detect patterns in quantifiers,
the first that enables complete
verification of the Padre algorithm used in the software of a Radio Telescope Pipeline. In addition, this
rewriting procedure
is described in full detail, and refers to a \lean correctness proof.

In principle, most quantifiers could be rewritten manually, but not those that \vercors must introduce automatically for parallel blocks and GPU kernels, which occur in most benchmarks. Among all experiments, \pvlinline|gemm_3| is the only one we could not verify. The rewriting procedure currently fails because its contract uses a linear pattern with the modulo (\pvlinline|%|) operator, which is not yet supported.

\paragraph{Conclusion.} In conclusion, the rewriting procedure is highly effective and allows us to automatically find triggers for most of the data-level parallel programs we considered. Additionally, the qualifiers of Section~\ref{sec:types}, especially when used in combination with the kernel extraction method discussed in the current section, are effective: they significantly reduce verification time and allow for the verification of otherwise unverifiable results.
\section{Related Work}\label{sec:related}

We divide our related work into two parts. First, we discuss some other approaches to support reasoning about quantifiers.
Dafny~\cite{leinoDafnyAutomaticProgram2010} defines symbolic equivalent functions for several arithmetic operators, which can be used in trigger patterns. This increases the expressiveness of trigger patterns. However, when an array is accessed in a different way than the \textit{precise} arithmetic form given in the trigger, it will not instantiate. To achieve successful verification in these cases, further complex reasoning about triggers is needed. Nonetheless, in the special case where arrays are always accessed in the exact same way, this is a valid alternative. Further, Dafny also supports a technique to automatically identify suitable triggers~\cite{leinoTriggerSelectionStrategies2016} 
In this paper our focus is complementary, as we try to increase the chances to find a matching pattern by rewriting the quantified expressions so that they can be matched with a larger number of terms in the specifications.

Second, we look at related work that identifies special restrictions, such as non-aliasing or immutability, by using types.
Chargu\'eraud and Pottier~\cite{chargueraudTemporaryReadOnlyPermissions2017} propose `temporary read-only permissions', an extension to standard separation-logic permissions that can mark memory locations as temporarily read-only without extra verification or specification overhead. This is similar to our \pvlinlineb|immutable| type modifier, except that once an array is marked immutable, it stays that way.
% They mention that the \pvlinline|const| modifier does not offer strong enough guarantees because, through aliasing, a \pvlinline|const| pointer can still be written to by another thread. We mitigate that by the \pvlinline|non-const| modifier, such that type checking can avoid these cases. \todo{Is that correct? We need to make sure the non-const is properly explained}
%
%agree with this observation; thus, in our work, when such behaviour could occur, these pointers should be annotated with. This means such pointers are treated like regular pointers, with permissions and all, such that the verification context can keep track of whether the array is ever updated.
%
In SYMPLAR~\cite{bierhoffAutomatedProgramVerification2011}, Bierhoff introduces \textit{symbolic permissions} as an alternative to fractional permissions. Users annotate types with \pvlinlineb|@Excl| (exclusive) or \pvlinlineb|@Imm| (immutable), and SYMPLAR checks these annotations. Our approach is similar but we integrate it with fractional permissions and still allow arrays with the same uniqueness type to alias, capturing a wider range of behaviours.
% With the \pvlinline|unique| type modifier of our work we do not enforce that two memory locations cannot alias, but we restrict the functions that use these types such that the function body cannot assume that memory locations might overlap and has to treat types of different uniqueness as separate locations.
%
Haack and Poll~\cite{haackTypeBasedObjectImmutability2009} present a type system for object immutability in Java. Our \pvlinlineb|immutable| reference corresponds to their \pvlinlineb|RdWr| type qualifier. The idea is the same, however we work in a verification setting. They also describe `Read-only references', where the object's state cannot be modified through this reference, analogous to \pvlinlineb|const| pointers in \C. A different encoding for read-only pointers seems feasible, but would still require permission bookkeeping in a separation-logic setting.
%!TEX root = main.tex

\section{Conclusions and Future Work}\label{sec:concl}
%This 
%paper presents several ways to better support optimised parallel code in a deductive verifier based on separation logic, in this case \vercors.

To improve the capabilities of deductive verifiers to verify data-level parallel programs, first of all, we introduced a rewriting
procedure for nested quantifiers, which improves trigger matching, and thereby greatly contributes to more efficient 
verification. We also introduced two special type qualifiers, to indicate that arrays are immutable or are not aliases of each other. The use of these qualifiers can be verified with type checking, and tend to greatly reduce verification time.
%, because it allows us to reduce of size of the symbolic state space that needs to be considered in verification.
The experimental results demonstrate the advantages of our optimisations. With these optimisations, we could significantly reduce verification times and verify previously unverifiable cases of a Radio Telescope Pipeline case study~\cite{vandenhaakVerifyingRadioTelescope2024}. In addition, employing all the suggested encoding optimisations makes the verification, on average, 9 times faster for the CLBlast GPU kernels, with maximum improvements of up to 150 times faster.

For future work, we see several possibilities for further improvements. 
%To correctly rewrite nested quantifiers, the verifier often needs to have additional annotations, such that we can statically determine properties like \pvlinlineb|n_k>0|. As an alternative,
For instance, we could add our rewriting procedure directly to \viper, allowing it to query the SMT solver immediately when Equations~\ref{eq:pat3-1}--\ref{eq:pat3-4} need to be checked. This has the advantage that the program state is already modelled at that point, and it enables other \viper front-ends to reuse our rewriting technique. % without implementing it themselves. However, it might contradict \viper's role as an \textit{intermediate} verification language, where adding complex rewrite passes might not be desirable.

%Additionally, beyond linearised array indexes, there might be many more instances for which we can find a bijective function that can be used to go from a multidimensional quantifier to a quantifier with one variable.

% For non-concrete bounds, this is still not enough. While we can provide the programs with accurately proven nonlinear information using the lemmas in Chapter~\ref{chp:padre}, we suspect that the larger program states hinder the underlying tools from retaining this information.

% As for the qualifiers themselves, the implementation of \pvlinline|const| qualifiers should be further improved. Currently, we can only cast or coerce a pointer to a \pvlinline|const| one by permanently removing permissions from the original pointer. This prevents writing to the original pointer again, which we need for this encoding. However, there are use cases where the original pointer is written to again, and this should invalidate the values of the \pvlinline|const| pointer.
% One direction would be to see if the work of Chargu\'eraud and Pottier~\cite{chargueraudTemporaryReadOnlyPermissions2017} we mentioned in the related work section could be used to implement what we want. However, this work is currently presented in the context of separation logic without fractional permissions and parallelism. Thus, we need to investigate how this translates towards the setting with fractional permissions and concurrency.

In addition, it would be interesting to investigate whether the \pvlinline|unique| type qualifier can be automatically applied in a
program based on the results of some static analysis. This would grant the programmer the benefits of this qualifier without burdening them
with the annotation task.

\newpage
\bibliographystyle{splncs04}
\bibliography{bibliography}

@inproceedings{WijsBos12,
	author = {A.J. Wijs and D. Bo\v{s}na\v{c}ki},
	booktitle = {Proceedings of the 19th International SPIN Workshop on Model Checking of Software (SPIN)},
	doi = {10.1007/978-3-642-31759-0_9},
	pages = {98-116},
	publisher = {Springer},
	series = {Lecture Notes in Computer Science},
	title = {Improving {GPU} {S}parse {M}atrix-{V}ector {M}ultiplication for {P}robabilistic {M}odel {C}hecking},
	volume = {7385},
	year = {2012}}

@inproceedings{GreLok11,
	author = {D. Grewe and A. Lokhmotov},
	booktitle = {Proceedings of the 4th {W}orkshop on General Purpose Processing on Graphics Processing Units (GPGPU)},
	doi = {10.1145/1964179.1964196},
	publisher = {ACM},
	title = {Automatically {G}enerating and {T}uning {GPU} {C}ode for {S}parse {M}atrix-{V}ector {M}ultiplication from a {H}igh-{L}evel {R}epresentation},
	year = {2011}}

@inproceedings{LiuTan12,
	author = {X. Liu and S. Tan and H. Wang},
	booktitle = {Proceedings of the 2012 Conference and Exhibition on Design, Automation \& Test in Europe (DATE)},
	doi = {10.1109/DATE.2012.6176615},
	pages = {852-857},
	publisher = {IEEE Computer Society},
	title = {Parallel {S}tatistical {A}nalysis of {A}nalog {C}ircuits by {GPU}-{A}ccelerated {G}raph-{B}ased {A}pproach},
	year = {2012}}

@inproceedings{deeplearning,
	author = {Q.V. Le and J. Ngiam and A. Coates and A. Lahiri and B. Prochnow and A.Y. Ng},
	booktitle = {Proceedings of the 28th International Conference on Machine Learning (ICML)},
	pages = {265-272},
	publisher = {Omnipress},
	title = {On {O}ptimization {M}ethods for {D}eep {L}earning},
	year = {2011}}

@inproceedings{BerBet11,
	author = {C. Bertolli and A. Betts and G. Mudalige and M. Giles and P. Kelly},
	booktitle = {In Proceedings of the 1st Workshop on Grids, Clouds and P2P Programming (CGWS)},
	doi = {10.1007/978-3-642-29737-3_22},
	pages = {191-200},
	publisher = {Springer},
	series = {Lecture Notes in Computer Science},
	title = {Design and {P}erformance of the {OP2} {L}ibrary for {U}nstructured {M}esh {A}pplications},
	volume = {7155},
	year = {2011}}

@inproceedings{WieSpri12,
	author = {S. Wienke and P. Springer and C. Terboven and D. Mey},
	booktitle = {Proceedings of the 18th European Conference on Parallel and Distributed Computing (EuroPar)},
	doi = {10.1007/978-3-642-32820-6_85},
	pages = {859-870},
	publisher = {Springer},
	series = {Lecture Notes in Computer Science},
	title = {{OpenACC} - {F}irst {E}xperiences with {R}eal-{W}orld {A}pplications},
	volume = {7484},
	year = {2012}}

@inproceedings{armborstVerCorsVerifierProgress2024,
	address = {Cham},
	author = {Armborst, Lukas and Bos, Pieter and {van den Haak}, Lars B. and Huisman, Marieke and Rubbens, Robert and {\c S}akar, {\"O}mer and Tasche, Philip},
	booktitle = {Computer {{Aided Verification}}},
	doi = {10.1007/978-3-031-65630-9_1},
	editor = {Gurfinkel, Arie and Ganesh, Vijay},
	isbn = {978-3-031-65630-9},
	langid = {english},
	pages = {3--18},
	publisher = {Springer Nature Switzerland},
	series = {Lecture {{Notes}} in {{Computer Science}}},
	shorttitle = {The {{VerCors Verifier}}},
	title = {The {{VerCors Verifier}}: {{A Progress Report}}},
	volume = {14682},
	year = 2024}

@inproceedings{bierhoffAutomatedProgramVerification2011,
	address = {Portland Oregon USA},
	author = {Bierhoff, Kevin},
	booktitle = {Proceedings of the 10th {{SIGPLAN}} Symposium on {{New}} Ideas, New Paradigms, and Reflections on Programming and Software},
	doi = {10.1145/2048237.2048242},
	isbn = {978-1-4503-0941-7},
	langid = {english},
	month = oct,
	pages = {19--32},
	publisher = {ACM},
	shorttitle = {Automated Program Verification Made {{SYMPLAR}}},
	title = {Automated Program Verification Made {{SYMPLAR}}: Symbolic Permissions for Lightweight Automated Reasoning},
	urldate = {2025-05-10},
	year = 2011}

@article{blomSpecificationVerificationGPGPU2014,
	author = {Blom, Stefan and Huisman, Marieke and Mihel{\v c}i{\'c}, Matej},
	doi = {10.1016/j.scico.2014.03.013},
	ids = {blomSpecificationVerificationGPGPU2014a},
	issn = {01676423},
	journal = {Science of Computer Programming},
	langid = {english},
	month = dec,
	pages = {376--388},
	title = {Specification and Verification of {{GPGPU}} Programs},
	urldate = {2020-06-16},
	volume = {95},
	year = 2014}

@inproceedings{brookesSemanticsConcurrentSeparation2004,
	address = {Berlin, Heidelberg},
	author = {Brookes, Stephen},
	booktitle = {{{CONCUR}} 2004 - {{Concurrency Theory}}},
	doi = {10.1007/978-3-540-28644-8_2},
	editor = {Gardner, Philippa and Yoshida, Nobuko},
	isbn = {978-3-540-28644-8},
	langid = {english},
	pages = {16--34},
	publisher = {Springer},
	title = {A {{Semantics}} for {{Concurrent Separation Logic}}},
	year = 2004}

@incollection{chargueraudTemporaryReadOnlyPermissions2017,
	address = {Berlin, Heidelberg},
	author = {Chargu{\'e}raud, Arthur and Pottier, Fran{\c c}ois},
	booktitle = {Programming {{Languages}} and {{Systems}}},
	doi = {10.1007/978-3-662-54434-1_10},
	editor = {Yang, Hongseok},
	isbn = {978-3-662-54433-4 978-3-662-54434-1},
	langid = {english},
	pages = {260--286},
	publisher = {Springer Berlin Heidelberg},
	title = {Temporary {{Read-Only Permissions}} for {{Separation Logic}}},
	urldate = {2025-05-04},
	volume = {10201},
	year = 2017}

@article{detlefsSimplifyTheoremProver2005,
	author = {Detlefs, David and Nelson, Greg and Saxe, James B.},
	doi = {10.1145/1066100.1066102},
	issn = {0004-5411, 1557-735X},
	journal = {Journal of the ACM},
	langid = {english},
	month = may,
	number = {3},
	pages = {365--473},
	shorttitle = {Simplify},
	title = {Simplify: A Theorem Prover for Program Checking},
	urldate = {2025-05-07},
	volume = {52},
	year = 2005}

@inproceedings{eilersFifteenYearsViper2025,
	author = {Eilers, Marco and Schwerhoff, Malte and Summers, Alexander J. and M{\"u}ller, Peter},
	booktitle = {Computer {{Aided Verification}} ({{CAV}})},
	title = {Fifteen {{Years}} of {{Viper}}},
	urldate = {2025-06-11},
	year = 2025}

@incollection{leinoDafnyAutomaticProgram2010,
	address = {Berlin, Heidelberg},
	author = {Leino, K. Rustan M.},
	booktitle = {Logic for {{Programming}}, {{Artificial Intelligence}}, and {{Reasoning}}},
	doi = {10.1007/978-3-642-17511-4_20},
	editor = {Clarke, Edmund M. and Voronkov, Andrei},
	isbn = {978-3-642-17510-7 978-3-642-17511-4},
	langid = {english},
	pages = {348--370},
	publisher = {Springer Berlin Heidelberg},
	shorttitle = {Dafny},
	title = {Dafny: {{An Automatic Program Verifier}} for {{Functional Correctness}}},
	urldate = {2025-03-27},
	volume = {6355},
	year = 2010}

@inproceedings{leinoTriggerSelectionStrategies2016,
	address = {Cham},
	author = {Leino, K. R. M. and {Pit-Claudel}, Cl{\'e}ment},
	booktitle = {Computer Aided Verification},
	editor = {Chaudhuri, Swarat and Farzan, Azadeh},
	isbn = {978-3-319-41528-4},
	pages = {361--381},
	publisher = {Springer International Publishing},
	title = {Trigger Selection Strategies to Stabilize Program Verifiers},
	year = 2016}

@inproceedings{mouraLean4Theorem2021,
	address = {Cham},
	author = {de Moura, Leonardo and Ullrich, Sebastian},
	booktitle = {Automated {{Deduction}} -- {{CADE}} 28},
	doi = {10.1007/978-3-030-79876-5_37},
	editor = {Platzer, Andr{\'e} and Sutcliffe, Geoff},
	isbn = {978-3-030-79876-5},
	langid = {english},
	pages = {625--635},
	publisher = {Springer International Publishing},
	title = {The {{Lean}} 4 {{Theorem Prover}} and {{Programming Language}}},
	year = 2021}

@incollection{mullerAutomaticVerificationIterated2016,
	address = {Cham},
	author = {M{\"u}ller, Peter and Schwerhoff, Malte and Summers, Alexander J.},
	booktitle = {Computer {{Aided Verification}}},
	doi = {10.1007/978-3-319-41528-4_22},
	editor = {Chaudhuri, Swarat and Farzan, Azadeh},
	isbn = {978-3-319-41527-7 978-3-319-41528-4},
	langid = {english},
	pages = {405--425},
	publisher = {Springer International Publishing},
	title = {Automatic {{Verification}} of {{Iterated Separating Conjunctions Using Symbolic Execution}}},
	urldate = {2024-06-19},
	volume = {9779},
	year = 2016}

@inproceedings{mullerViperVerificationInfrastructure2016,
	address = {Berlin, Heidelberg},
	author = {M{\"u}ller, Peter and Schwerhoff, Malte and Summers, Alexander J.},
	booktitle = {Verification, {{Model Checking}}, and {{Abstract Interpretation}}},
	doi = {10.1007/978-3-662-49122-5_2},
	editor = {Jobstmann, Barbara and Leino, K. Rustan M.},
	isbn = {978-3-662-49122-5},
	langid = {english},
	pages = {41--62},
	publisher = {Springer},
	shorttitle = {Viper},
	title = {Viper: {{A Verification Infrastructure}} for {{Permission-Based Reasoning}}},
	year = 2016}

@inproceedings{nugterenCLBlastTunedOpenCL2018,
	address = {New York, NY, USA},
	author = {Nugteren, Cedric},
	booktitle = {Proceedings of the {{International Workshop}} on {{OpenCL}}},
	doi = {10.1145/3204919.3204924},
	isbn = {978-1-4503-6439-3},
	month = may,
	pages = {1--10},
	publisher = {Association for Computing Machinery},
	series = {{{IWOCL}} '18},
	shorttitle = {{{CLBlast}}},
	title = {{{CLBlast}}: {{A Tuned OpenCL BLAS Library}}},
	urldate = {2026-01-11},
	year = 2018}

@article{ragan-kelleyHalideDecouplingAlgorithms2017,
	author = {{Ragan-Kelley}, Jonathan and Adams, Andrew and Sharlet, Dillon and Barnes, Connelly and Paris, Sylvain and Levoy, Marc and Amarasinghe, Saman and Durand, Fr{\'e}do},
	doi = {10.1145/3150211},
	issn = {0001-0782},
	journal = {Communications of the ACM},
	month = dec,
	number = {1},
	pages = {106--115},
	shorttitle = {Halide},
	title = {Halide: Decoupling Algorithms from Schedules for High-Performance Image Processing},
	urldate = {2023-01-20},
	volume = {61},
	year = 2017}

@phdthesis{schwerhoffAdvancingAutomatedPermissionBased2016,
	author = {Schwerhoff, Malte H.},
	copyright = {http://rightsstatements.org/page/InC-NC/1.0/},
	doi = {10.3929/ethz-a-010835519},
	langid = {english},
	school = {ETH Zurich},
	title = {Advancing {{Automated}}, {{Permission-Based Program Verification Using Symbolic Execution}}},
	type = {Doctoral {{Thesis}}},
	urldate = {2025-04-05},
	year = 2016}

@inproceedings{vandenhaakHaliVerDeductiveVerification2024,
	address = {Cham},
	author = {{van den Haak}, Lars B. and Wijs, Anton J. and Huisman, Marieke and {van den Brand}, Mark G. J.},
	booktitle = {Tools and {{Algorithms}} for the {{Construction}} and {{Analysis}} of {{Systems}}},
	doi = {10.1007/978-3-031-57256-2_4},
	editor = {Finkbeiner, Bernd and Kov{\'a}cs, Laura},
	isbn = {978-3-031-57256-2},
	langid = {english},
	pages = {71--89},
	publisher = {Springer Nature Switzerland},
	series = {Lecture {{Notes}} in {{Computer Science}}},
	shorttitle = {{{HaliVer}}},
	title = {{{HaliVer}}: {{Deductive Verification}} and {{Scheduling Languages Join Forces}}},
	volume = {14572},
	year = 2024}

@inproceedings{vandenhaakVerifyingRadioTelescope2024,
	address = {Cham},
	author = {{van den Haak}, Lars B. and Wijs, Anton J. and Huisman, Marieke and {van den Brand}, Mark G. J.},
	booktitle = {Formal {{Methods}} for {{Industrial Critical Systems}}},
	doi = {10.1007/978-3-031-68150-9_9},
	editor = {Haxthausen, Anne E. and Serwe, Wendelin},
	isbn = {978-3-031-68150-9},
	langid = {english},
	pages = {152--169},
	publisher = {Springer Nature Switzerland},
	series = {Lecture {{Notes}} in {{Computer Science}}},
	shorttitle = {Verifying a~{{Radio Telescope Pipeline Using HaliVer}}},
	title = {Verifying a~{{Radio Telescope Pipeline Using HaliVer}}: {{Solving Nonlinear}} and~{{Quantifier Challenges}}},
	volume = {14952},
	year = 2024}

@incollection{wolfGobraModularSpecification2021,
	address = {Cham},
	author = {Wolf, Felix A. and Arquint, Linard and Clochard, Martin and Oortwijn, Wytse and Pereira, Jo{\~a}o C. and M{\"u}ller, Peter},
	booktitle = {Computer {{Aided Verification}}},
	doi = {10.1007/978-3-030-81685-8_17},
	editor = {Silva, Alexandra and Leino, K. Rustan M.},
	isbn = {978-3-030-81684-1 978-3-030-81685-8},
	langid = {english},
	pages = {367--379},
	publisher = {Springer International Publishing},
	shorttitle = {Gobra},
	title = {Gobra: {{Modular Specification}} and {{Verification}} of {{Go Programs}}},
	urldate = {2026-01-05},
	volume = {12759},
	year = 2021}

@article{hillisDataParallelAlgorithms1986,
  title = {Data Parallel Algorithms},
  author = {Hillis, W. Daniel and Steele, Guy L.},
  year = 1986,
  month = dec,
  journal = {Commun. ACM},
  volume = {29},
  number = {12},
  pages = {1170--1183},
  issn = {0001-0782},
  doi = {10.1145/7902.7903},
  urldate = {2026-01-26}
}

@incollection{haackTypeBasedObjectImmutability2009,
  title = {Type-{{Based Object Immutability}} with {{Flexible Initialization}}},
  booktitle = {{{ECOOP}} 2009 -- {{Object-Oriented Programming}}},
  author = {Haack, Christian and Poll, Erik},
  editor = {Drossopoulou, Sophia},
  year = 2009,
  volume = {5653},
  pages = {520--545},
  publisher = {Springer Berlin Heidelberg},
  address = {Berlin, Heidelberg},
  doi = {10.1007/978-3-642-03013-0_24},
  urldate = {2026-01-26},
  copyright = {http://www.springer.com/tdm},
  isbn = {978-3-642-03012-3 978-3-642-03013-0},
  langid = {english}
}
\newpage
\appendix

% \section{Determining the truth of condition~\ref{eq:pat3-4} }
\section{Lean proof of Theorem~\ref{thm:eq-quant}}\label{sec:appendix-lean-proof}
The complete Lean proof can be found in the accompanying artefact: \url{https://github.com/cav2026-anonymous/Scalable-Deductive-Verification-of-Data-Level-Parallel-Programs/tree/main/LeanProof}.
We have implemented \pvlinline|f|, \pvlinline|f_inv| ($f^{-1}$), \pvlinline|base|, \pvlinline|off|, starting domain \pvlinline|X| and resulting domain \pvlinline|Y| to closely match the definitions given in this section in the file \verb|Definitions.lean|. The properties \ref{eq:pat3-1}, \ref{eq:pat3-2}, and \ref{eq:pat3-3} are modelled in the definition \verb|Props|. The property \ref{eq:pat3-4} is added to the definition of \verb|X|.

One important difference from our Lean proof is that we index slightly differently, since we modelled $(x_1, \ldots, x_k) \in \mathds{Z}$ as vector \verb|xs| (\verb|Vector Int k|) and it was easier to model everything reversed. Thus, $x_1$ corresponds to the last element of the vector \verb|xs.get (k-1)| and $x_k$ corresponds to \verb|xs.get 0|. This also means that the definition of $\base_i$ is reversed.

Eventually, the results can be found in the file \verb|Results.lean|. Theorem \verb | f_inv_on'| proves that $f$ and $f^{-1}$ are inverses on the sets $X$ and $Y$. Whilst \verb|f_bij_on'| proves that $f$ is a bijection between $X$ and $Y$.

Lastly, the result is in theorem \verb|equiv_quantifier'|, which proves Theorem~\ref{thm:eq-quant}.

\section{Determining the constraint of Equation~\ref{eq:pat3-4}}\label{sec:appendix-constraint}
For clarity, we repeat the constraint again.
\begin{equation*}
\forall i, 1 \leq i < k \Rightarrow \forall(x_1, \ldots, x_k) \in X, 
  \sum_{j=1}^{i} |a_j|\cdot (x_j-\xmin_j) < |a_{i+1}|
\end{equation*}

Determining whether this equation holds is generally difficult, since most quantifiers will not include this equation \textit{ specifically} as part of their domain. Especially if an $a_i$ is not a concrete number, these are nonlinear equations that we cannot easily solve. Therefore, we looked for additional components of the quantifier for which we could solve (part) of this equation more easily.

Internally in \vercors, we check equation~\ref{eq:pat3-4} by starting at the lowest $i=1$ and iterating up to $k$. Suppose that you know that up to $i-1$ we have proven the equation and now we want to prove that it holds for $i$.

We consider two cases in which we can prove this. In both cases $x_i$ should have an upper bound $\xmax_i$, such that $n_i$ is defined and $n_i>0$. 
The first case is that $n_i\cdot a_i = a_{i+1}$ holds, the second is that $|a_{i}|\cdot n_i \leq |a_{i+1}|$ holds. In the following lemma, we prove that this is correct.

\begin{lemma}\label{lemma:2patterns}
Suppose that for any integers $i$ and $k$, $1 \leq i < k$, and all 
\\ $(x_1, \ldots, x_k) \in X$ it is proven that
\begin{equation}\label{eq:proof-pat-1}
\forall i' \in \mathbb{Z}, 1 \leq i' < i \Rightarrow
  \sum_{j=1}^{i'} |a_j|\cdot (x_j-\xmin_j) < |a_{i'+1}|
\end{equation}
holds, we have that $a_i\neq0$, $n_i >0$, all $a_i$ have the same sign, and we have either that $n_i\cdot a_i = a_{i+1}$ or that $|a_{i}|\cdot n_i \leq |a_{i+1}|$
holds. Then
\begin{equation}\label{eq:proof-pat-2}
\sum_{j=1}^{i} |a_j|\cdot (x_j-\xmin_j) < |a_{i+1}|
\end{equation}
is true.
\end{lemma}
\begin{proof}
\noindent
\textbf{Case $i=1$ and $n_i\cdot a_i = a_{i+1}$}:\\
% We need to prove that $|a_1|\cdot(x_1 - \xmin_1) < |a_2|$.
We have $x_1 < \xmax_1$, which implies that $x_1-\xmin_1 < \xmax_1-\xmin_1=n_1$, which in turn implies $|a_1|\cdot(x_1-\xmin_1) < n_1\cdot|a_1| = |n_1 \cdot a_1|$.

We have $a_1 \cdot n_1 = a_{2}$. Thus, we then have
$|a_1|\cdot(x_1-\xmin_1) < |a_2|$, which was what we wanted to prove for \ref{eq:proof-pat-2}.

\noindent
\textbf{Case $i>1$ and $n_i\cdot a_i = a_{i+1}$}:\\
% We only have to consider the case that $i<k$. For bigger $i$, the claim is vacuously true.
That is, we want to prove the following. $$\sum_{j=1}^{i} |a_j| \cdot (x_j-\xmin_j) < |a_{i+1}|$$

We have that $x_{i} < \xmax_{i}$, which implies $x_{i}-\xmin_{i} \leq \xmax_{i}-\xmin_{i}-1=n_{i}-1$, which in turn implies $|a_{i}|\cdot(x_{i}-\xmin_{i}) \leq (n_{i}-1)\cdot|a_{i}|=|n_{i}\cdot a_{i}|-|a_{i}|$.

We have $a_{i} \cdot n_{i} = a_{i+1}$. Thus, we then have
\begin{equation}\label{eq:fact1}
|a_{i}|\cdot(x_{i}-\xmin_{i}) \leq |a_{i+1}|-|a_{i}|  
\end{equation}
We rewrite the proof goal.
\begin{align*}
&\sum_{j=1}^{i-1} |a_j| \cdot (x_j-\xmin_j) \\
=&|a_{i}|\cdot(x_{i}-\xmin_{i})+ \sum_{j=1}^{i-1} |a_j| \cdot  (x_j-\xmin_j) \\
\end{align*}
Using equations~\ref{eq:fact1}, and \ref{eq:proof-pat-1} by filling in $i'=i-1$, we get
\begin{align*}
&\sum_{j=1}^{i-1} |a_j| \cdot (x_j-\xmin_j) < |a_{i+1}|-|a_{i}| + |a_{i}| = |a_{i+1}|
\end{align*}
This is exactly equation~\ref{eq:proof-pat-2} which was the proof goal.

\noindent
\textbf{Case $i=1$ and $|a_{i}|\cdot n_i \leq |a_{i+1}|$}:\\
% We need to prove that $|a_1|\cdot(x_1 - \xmin_1) < |a_2|$.
We have $x_1 < \xmax_1$, which implies that $x_1-\xmin_1 < \xmax_1-\xmin_1=n_1$, which in turn implies $|a_1|\cdot(x_1-\xmin_1) < |a_1| \cdot n_1$.

We have $|a_1| \cdot n_1 \leq |a_{2}|$. Thus, we then have
$|a_1|\cdot(x_1-\xmin_1) < |a_1|\cdot n_1 \leq |a_2|$, which was what we wanted to prove for \ref{eq:proof-pat-2}.

\noindent
\textbf{Case $i>1$ and $|a_{i}|\cdot n_i \leq |a_{i+1}|$}:\\
We have that $x_{i} < \xmax_{i}$, which implies $x_{i}-\xmin_{i} \leq \xmax_{i}-\xmin_{i}-1=n_{i}-1$, which in turn implies $|a_{i}|\cdot(x_{i}-\xmin_{i}) \leq (n_{i}-1)\cdot|a_{i}|=|n_{i}\cdot a_{i}|-|a_{i}|$.

We have $|a_{i}| \cdot n_{i} \leq |a_{i+1}|$. Thus, we then have
\begin{equation}\label{eq:fact2}
|a_{i}|\cdot(x_{i}-\xmin_{i}) \leq |a_{i+1}|-|a_{i}|
\end{equation}
We rewrite the proof goal.
\begin{align*}
&\sum_{j=1}^{i} |a_j| \cdot (x_j-\xmin_j)\\
=&|a_{i}|\cdot(x_{i}-\xmin_{i})+ \sum_{j=1}^{i-1} |a_j| \cdot (x_j-\xmin_j)\\
\end{align*}
Using equations~\ref{eq:fact2}, and \ref{eq:proof-pat-1} by filling in $i'=i-1$ we get
\begin{align*}
\sum_{j=1}^{i} |a_j| \cdot (x_j-\xmin_j) < |a_{i+1}|-|a_{i}| + |a_{i}| = |a_{i+1}|
\end{align*}
This is exactly equation~\ref{eq:proof-pat-2} which was the proof goal.

Since either $i=1$ or $i>1$ holds and either $|a_{i}|\cdot n_i \leq |a_{i+1}|$ or $n_i\cdot a_i = a_{i+1}$ holds, we covered all the cases and proved in all cases that equation \ref{eq:proof-pat-2} holds.
\end{proof}

Thus, with the above lemma, it means that for each $i$ we can first check if $n_i\cdot a_i = a_{i+1}$ or $|a_{i}|\cdot n_i \leq |a_{i+1}|$ holds to determine the constraint of equation~\ref{eq:pat3-4}. Then only in the third fallback case do we ask our symbolic evaluator if equation~\ref{eq:pat3-4} holds directly for a specific $i$.

\end{document}